\documentclass[11pt]{amsart}
\usepackage{fullpage}
\usepackage[utf8]{inputenc}
\usepackage{graphics}
\usepackage{amsmath, amsfonts,graphicx,cite,comment,amsthm}
\usepackage{subcaption} 

\usepackage{mathtools}
\usepackage{color}

        \definecolor{pink}{rgb}{1,0,1}
	
        \definecolor{purple}{rgb}{0.4,0.2,1}
        
        \newcommand{\mca}{\operatorname{MCA}}

\newtheorem{theorem}{Theorem}
\newtheorem{defn}{Definition}
\newtheorem{prop}{Proposition} 
\newtheorem{cor}{Corollary}

\title[Microbial biodiversity and phenotypic variability]{Biodiversity of marine microbes is safeguarded by phenotypic variability in ecological traits}

\author{S. Menden-Deuer$^{1}$ \\ M. Nursultanov$^{2}$ \\  S. Collins$^{3}$ \\ T. Rynearson$^{1}$\\  J. Rowlett$^{4,5, \ast}$ }

\date{}

\begin{document}

\maketitle

\noindent{} 1. University of Rhode Island;

\noindent{} 2. University of Sydney;

\noindent{} 3. Edinburg University; 

\noindent{} 4. Chalmers University; 

\noindent{} 5. University of Gothenburg.

\noindent{} $\ast$ Corresponding author; e-mail: julie.rowlett@chalmers.se .


\bigskip


\bigskip


\bigskip
\section*{Abstract}
Why, contrary to theoretical predictions, do marine microbe communities harbor tremendous phenotypic heterogeneity?  How can so many marine microbe species competing in the same niche coexist?  We discovered a unifying explanation for both phenomena by investigating a non-cooperative game that interpolates between individual-level competitions and species-level outcomes.  We identified all equilibrium strategies of the game.  These strategies are characterized by maximal phenotypic heterogeneity.  They are also neutral towards each other in the sense that an unlimited number of species can co-exist while competing according to the equilibrium strategies.  Whereas prior theory predicts that natural selection would minimize trait variation around an optimum value, here we obtained a rigorous mathematical proof that species with maximally variable traits are those that endure.  This discrepancy may reflect a disparity between predictions from models developed for larger organisms in contrast to our microbe-centric model.  Rigorous mathematics proves that phenotypic heterogeneity is itself a mechanistic underpinning of microbial diversity.  This discovery has fundamental ramifications for microbial ecology and may represent an adaptive reservoir sheltering biodiversity in changing environmental conditions.

\section*{Introduction}
Marine planktonic microbes make Earth habitable by collectively generating as much organic matter and oxygen as all terrestrial plants combined \cite{falkowski2008}. The ecological, biogeochemical and economical importance of marine microbes is rooted in their vast species and metabolic diversity \cite{falkowski2008, worden2015}.  Recent estimates of marine microbial species diversity exceed 200,000 species in the plankton \cite{devargas2015, sunagawa2015}. At all levels of taxonomy, from species to intra-strain comparisons, there exists a tremendous and theoretically inexplicable reservoir of variability in genetic, physiological, behavioral and morphological characteristics \cite{ahlgrenrocap, ahlgren2006, johnson2006, schaum2012, boyd2013, hutchins2013, kashtan2014, harvey2015, mdmontalbano, sohm2016, godherynearson, wolf2017, olofsson2018}.  The maintenance of such extraordinary diversity and persistent co-existence of planktonic microbes in a putatively isotropic environment represents a long-standing scientific enigma \cite{hutchinson1961} that has yielded a phenomenologically sound explanation \cite{huismanweissing} but to date lacks a general, mechanistic explanation.

While community and intraspecific diversity is usually framed in terms of genotype diversity, genetically-identical cells often have important phenotypic differences in homogeneous environments \cite{fontana2019, heyse2019, mizrachi2019}, a phenomenon called phenotypic heterogeneity \cite{ackermann2015} or intra-genotypic variability \cite{bruijning2020}. Although phenotypic heterogeneity is a key attribute of microbes, it is not typically examined as a potential driver of microbial diversity or species persistence. Importantly, phenotypic heterogeneity can be acted on by natural selection; it is heritable, can be altered by genetic change, and can influence survival and fitness \cite{raj2008, vineyreece, fontana2019}. Previous treatments of phenotypic heterogeneity focus on how phenotypic heterogeneity can be generated \cite{ackermann2015, calabrese2019}, how this characteristic can allow populations to use bet-hedging to persist in heterogenous environments \cite{ackermann2015}, to support niche partitioning \cite{huismanweissing}, or to facilitate intra-specific cooperation \cite{westcooper}. However, the role of phenotypic heterogeneity in maintaining species diversity remains unexplored. We propose that phenotypic heterogeneity and incorporation of cell-cell interactions is key to understanding coexistence in microbial communities and explains how large numbers of species can coexist, even in unstructured environments.

Here we examined how phenotypic heterogeneity affects population survival and coexistence. This corresponds to a scenario where all competing species are reasonably well-adapted to their environment. We leveraged game theory \cite{cowden2012} as a tool to explore the consequences of phenotypic heterogeneity for the outcomes of cell-cell competitions at the individual level and the resulting implications for the persistence of populations.  A non-cooperative game for competition between microbe species was introduced in \cite{mdr2014} and \cite{mdr2018}.  In that model, a 
player in the game-theoretic-sense  represents a species that is comprised of many individual microbes. This provides a crucial link between individual-level interactions and species-level repercussions. Importantly, the approach taken specifically formulated a model to reflect the characteristics of asexually reproducing microbes, and thus departs from approaches that utilize models designed for multicellular organisms to understand microbial ecology and evolution.  In \cite{mdr2018}, competitions between pairs of species were analyzed, and it was shown that there are no evolutionary stable strategies.  That result indicated that optimal trait distributions for microbe species might not be centered around an optimal value, but instead, that perhaps microbe species may be characterized by more variable trait distributions.  However, the model in \cite{mdr2018} did not analyze competition between arbitrary and fluctuating numbers of competing species, and in particular, did not identify the best strategies to promote survival and co-existence.  Here, we generalize the model in \cite{mdr2018} to create a non-cooperative game in which arbitrary and fluctuating numbers of microbe species compete.  We identify all equilibrium strategies.  These equilibrium strategies have two key characteristics: maximal phenotypic heterogeneity and neutrality towards all other strategies that are equally cumulatively fit.  


\section*{Methods:  A microbe centric competition model that links individual-level interactions to population-level consequences.}

Taking a microbe-centric point of view, we assume that species are represented by many individual, closely-related cells.   These individuals compete for limited resources.  The outcome of competition between individuals, no matter how biologically complex, can be reduced to three possibilities:  win, lose, or draw.  Since microorganisms reproduce asexually, success or failure in competition corresponds directly to population increase or decrease, respectively, as shown in Fig \ref{fig:1}.  

\begin{figure}[h] \centering \includegraphics[width=0.9\textwidth]{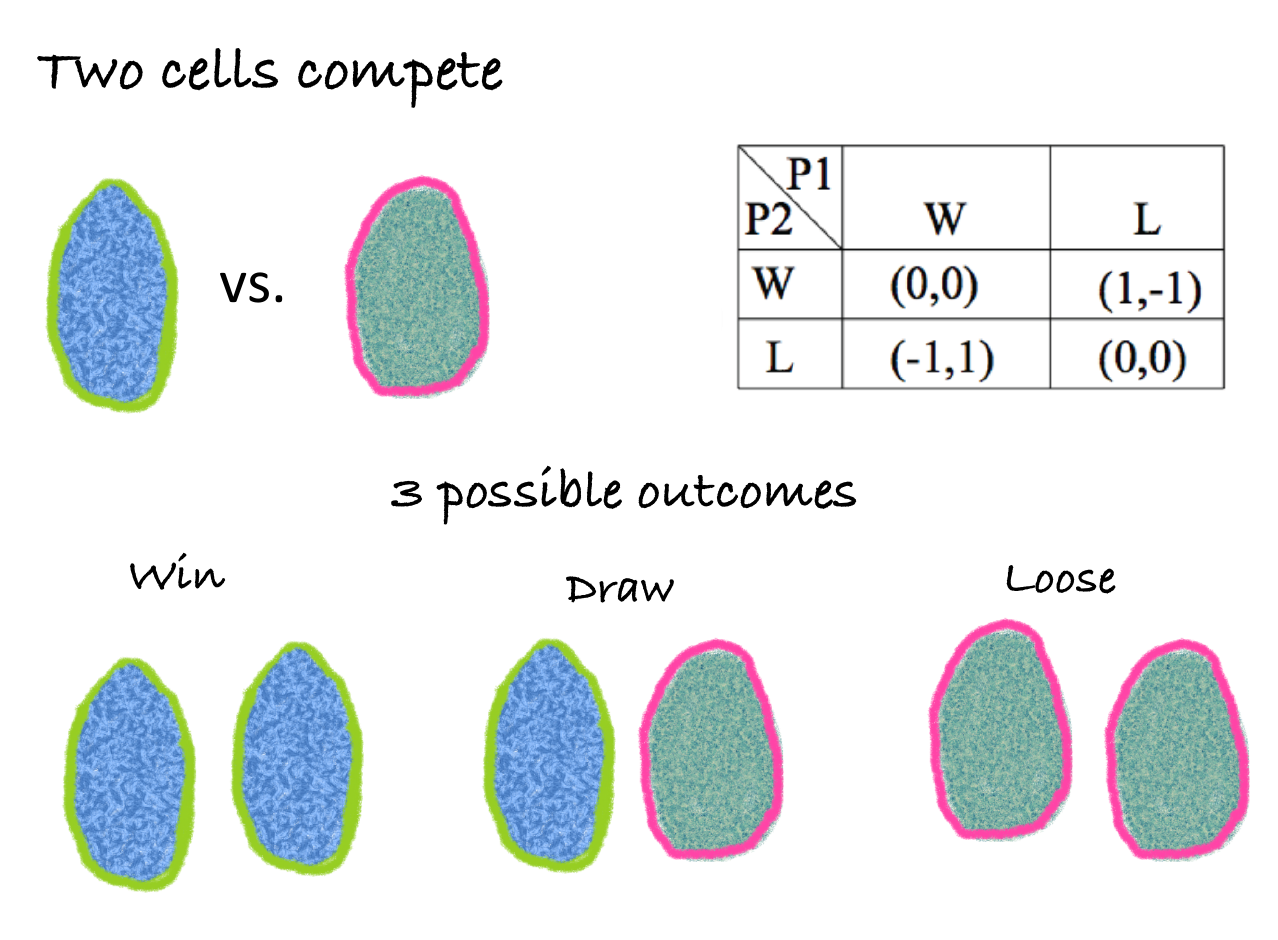}
\caption{{\bf Individuals compete for limited resources.} When two asexually reproducing individuals compete in a zero-sum game, corresponding to limited resources, there are three possible outcomes for each individual:  win=self-replicate; draw=maintain status quo for each competitor; lose=die/competitor replicates.  The table shows the pay-off function of individual-level competition.  A payoff of 1 corresponds to win, a payoff of -1 corresponds to lose, and a payoff of 0 corresponds to a draw.}
\label{fig:1}
\end{figure} 

At each round of competition, each individual is assigned a competitive ability according to the strategy of its species. Biologically, a strategy represents the probability distribution of competitive abilities (e.g. traits) for all individuals of the species.  Different trait distributions can be the result of random phenotypic noise, plasticity, or the genetic variation that inevitably exists in large, actively-dividing microbial populations \cite{rocha2018}.  Leveraging theoretical mathematics, we are able to simultaneously consider all traits and all types of distributions, rather than single traits with close to normal distributions as done previously \cite{barbarasdandrea}.  In the game theoretic sense, a trait distribution is a mixed strategy for the species, considered as a player in a non-cooperative game. Representing competitive ability as a distribution, rather than a mean value exemplifies our approach of incorporating phenotypic heterogeneity in the assessment of competition outcomes. 

\begin{figure} \centering \includegraphics[width=\textwidth]{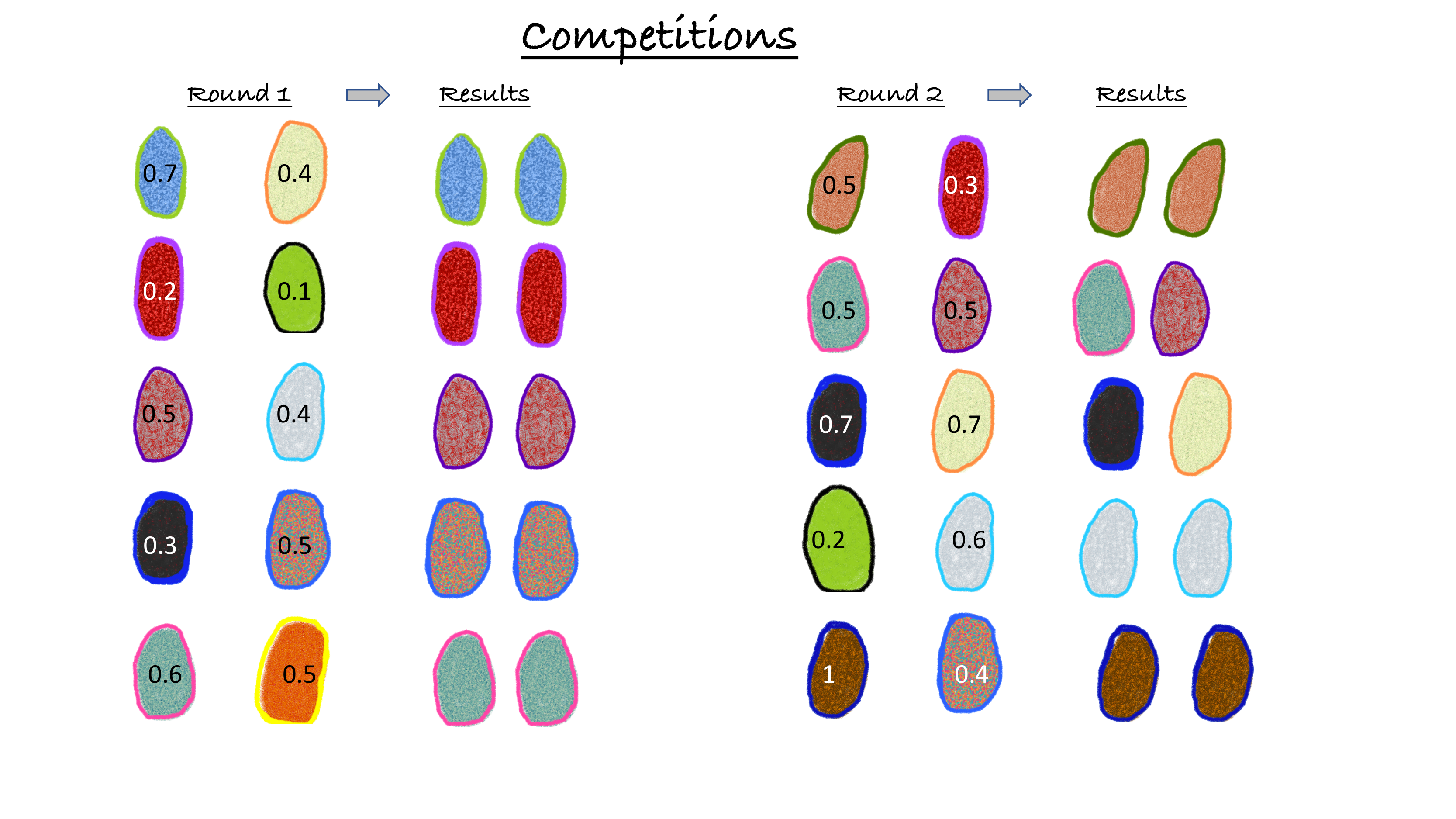}
\caption{{\bf Species comprised of individuals compete for limited resources.}  At each round of competition, diverse individuals compete and are assessed based on their relative competitive ability (values within cells). Different colors can represent different individuals in a cohort, which can represent a clone or species. Following empirical observations, the competitive ability of an individual is not necessarily constant from one round of competition to the next.  Moreover, if an individual wins a competition, then its descendant does not necessarily inherit the same competitive ability.  Rather, at each round of competition the competitive abilities of the individuals are assigned according to the strategy of the species. Thus, variability amongst individuals and the strategy of the species is preserved.}
\label{fig:2}
\end{figure} 

Game theory is uniquely suited to examine the effects of phenotypic heterogeneity on survival of competing microbial species because game theory evaluates population-level outcomes of individual-level interactions as depicted in Fig \ref{fig:2}. In our application of game theory, we define species-level strategies from which we derive the outcomes of individual-level competitions that ultimately determine the population-level survival and abundance. This approach is in contrast with population-level assessments, like Lotka-Volterra type models, where phenotypic heterogeneity is accounted for by using stochastic differential equations \cite{berryman1992}.  The game theoretic framework provides a quantitative means to evaluate the success of species competing according to their competitive strategies.  Below, we obtain tractable equations for each individual microbe, allowing us to mathematically represent individuals yet analyze all individuals in a population collectively. The resulting set of equations explicitly interpolates between microscopic or individual-interactions and macroscopic, or population-level consequences.  A key realization is that there are many unique strategies that all have an identical mean competitive ability. This characteristic builds a fundament of why many species with phenotypic heterogeneity can co-exist: because the competing individuals can yet differ.

\subsection*{The discrete and continuous models} 
We consider $n$ competing species and allow $n$ to vary.  Each species has a strategy, and we may use the same notation to denote both the species as well as its strategy.  Here we shall analyze two models:  discrete and continuous.  In the discrete model, there is a finite set of competitive abilities: 

\[ \left\{ \frac j M \right\}_{j=0} ^M,  \quad x_j := \frac j M.\] 
A strategy in the discrete model is a map 

\[ A : \left\{  x_j \right\}_{j=0} ^M \to [0, \infty)^{M+1}.\] 
We use $A$ to denote both the species and its strategy.  To reflect biological constraints, we set a constraint on the \em mean competitive ability (MCA) \em for all species.  

\begin{defn} \label{def:strategy_d} A \em strategy \em in the discrete case is a map $A : \left\{ x_j \right\}_{j=0} ^M \to [0, \infty)^{M+1}$ such that 
\begin{equation} \label{eq:mca_constraint_discrete} |A| := \sum_{j=0} ^M A(x_j) > 0 \textrm{ and }  \mca (A) := \frac{1}{|A|} \sum_{j=0} ^M x_j A (x_j) \leq \frac 1 2. \end{equation} 
\end{defn} 

Consequently, biologically we interpret $|A|$ as the population of $A$.  A species whose population is zero cannot compete; thus we only analyze strategies with a positive population, but we note that the population need not be integer-valued.  The strategy assigns competitive abilities to individuals in the sense that $ \frac{A(x_j)}{|A|}$ is the probability that a randomly selected individual from species $A$ has competitive ability equal to $x_j$.  In game theory, it may seem more natural to define the strategy instead as a map
\[ \widetilde A : \left\{ x_j \right\}_{j=0} ^M \to [0, 1]^{M+1} \textrm{ and require that } \sum_{j=0} ^M \widetilde A(x_j) = 1.\] 
This would then imply that all species have population equal to one, but we wish to allow species of different populations to compete, and therefore we do not make this normalization.  However, it is wholly equivalent to consider $\widetilde A := A/|A|$ as the strategy of the species with the population equal to $|A|$, requiring that $|A|>0$ but need not equal one, and with this normalization indeed 
\[ \widetilde A := \frac{A}{|A|} : \left\{ x_j \right\}_{j=0} ^M \to [0, 1]^{M+1}, \quad \sum_{j=0} ^M \widetilde A(x_j) = 1.\] 

One could imagine that the number of competitive abilities might change or possibly tend to infinity, and for this reason our analyses include a second model, \em the continuous model. \em   In the continuous model, the competitive abilities are selected from the entire range of real numbers between $0$ and $1$.  A species is associated with a continuous non-negative function defined on the unit interval, that is used to define the strategy of the species.  

\begin{defn} \label{def:strategy_c} A \em strategy \em in the continuous model is a map $f: [0,1] \to [0, \infty)$ that satisfies
\begin{equation} \label{eq:mca_constraint_cont} 
F(1) := \int_0 ^1 f(x) dx > 0 \textrm{ and }  \mca(f) := \frac{1}{F(1)} \int_0 ^1 x f(x) dx \leq \frac 1 2.  \end{equation} 
\end{defn} 
Similar to the discrete model, a species whose population is zero cannot compete, so we only analyze those species that have positive, but not necessarily integer-valued, populations.  The continuous model is obtained from the discrete model by letting the number of competitive abilities $M \to \infty$.  In game theory, it may seem more natural to define the strategy instead as the function
\[ \widetilde f(x) := \frac{f(x)}{F(1)} \implies \widetilde f(x) \geq 0 \forall x, \quad \int_0 ^1 \widetilde f(x) dx = 1.\] 
The function $\widetilde f$ is then the probability distribution of competitive abilities of the species, but with this normalization, again 
all species would have populations equal to one, and we wish to allow species of different populations to compete.  Consequently, we do not make this normalization, but it is mathematically equivalent to consider $\widetilde f$ as the strategy of the species with the population to $F(1)$, requiring only that $F(1)>0$ but need not equal one.

The change in population of the species is determined by the game theoretic payoffs.  These payoffs are defined by the expected value in competition when all species simultaneously compete.  Consequently, in the discrete case, we define the payoff to species $A_k$ to be

\begin{equation} \label{eq:def_payoff_d} \wp (A_k; A_1, \ldots, A_{k-1}, A_{k+1}, A_n) := \frac{1}{\sum_{\ell=1} ^n |A_\ell|} \sum_{j=0} ^M A_k (x_j) \left [ \sum_{i<j} \sum_{\ell=1} ^n A_\ell (x_i) - \sum_{i>j} \sum_{\ell=1} ^n A_\ell (x_i) \right ]. \end{equation} 
Above and throughout this work, an empty sum is taken to be equal to zero.   Here, species are competing both internally and externally.  Due to the zero sum dynamic, however, internal competition within a species does not affect its population.  Consequently, these payoffs 
show how the species' populations increase, decrease, or remain stable when the species compete according to their strategies.  
The payoffs in the continuous case are defined in an analogous way, so that the payoff to species $f_k$ is 

\begin{equation} \label{eq:def_payoff_c} \wp (f_k; f_1, \ldots, f_{k-1}, f_{k+1}, \ldots, f_n) = \frac{1}{\sum_{\ell=1} ^n F_\ell (1)} \int_0 ^1 f_k (x) \left[ \int_0 ^x\sum_{\ell=1} ^n f_\ell (t) - \int_x ^1 \sum_{\ell=1} ^n f_\ell (t) \right] dx. \end{equation}

Having obtained this unique set of equations for all individual microbes connecting individual-level interactions to population-level consequences, we are poised to mathematically analyze all strategies.  What are the best strategies?  A fundamental concept in non-cooperative game theory is an equilibrium strategy:  no one player can increase their payoff if only that one player changes their strategy whilst other players keep their strategies fixed \cite{nash1951}.

\begin{defn} \label{def:eq_str} In the discrete case, an equilibrium strategy is a set of strategies $A_1, \ldots, A_n$ as in Definition \ref{def:strategy_d} such that for all $k=1, \ldots, n$ we have 
\[ \wp (A_k; A_1, \ldots, A_{k-1}, A_{k+1}, \ldots A_n) \geq \wp (B; A_1, \ldots, A_{k-1}, A_{k+1}, \ldots, A_n),\] 
for all strategies $B$ as in Definition \ref{def:strategy_d}.  In the continuous case, an equilibrium strategy is a set of strategies $f_1, \ldots f_n$ as in Definition \ref{def:strategy_c} such that for all $k=1, \ldots, n$ we have 
\[ \wp (f_k; f_1 \ldots, f_{k-1}, f_{k+1}, \ldots f_n) \geq \wp (g; f_1 \ldots, f_{k-1}, f_{k+1}, \ldots, f_n),\] 
for all strategies $g$ as in Definition \ref{def:strategy_c}. 
\end{defn}


\subsection*{Simulations for the discrete model} \label{s:simulations} 
\begin{figure}[h]
\centering
   \begin{subfigure}{0.32\linewidth} \centering
     \includegraphics[width=\textwidth]{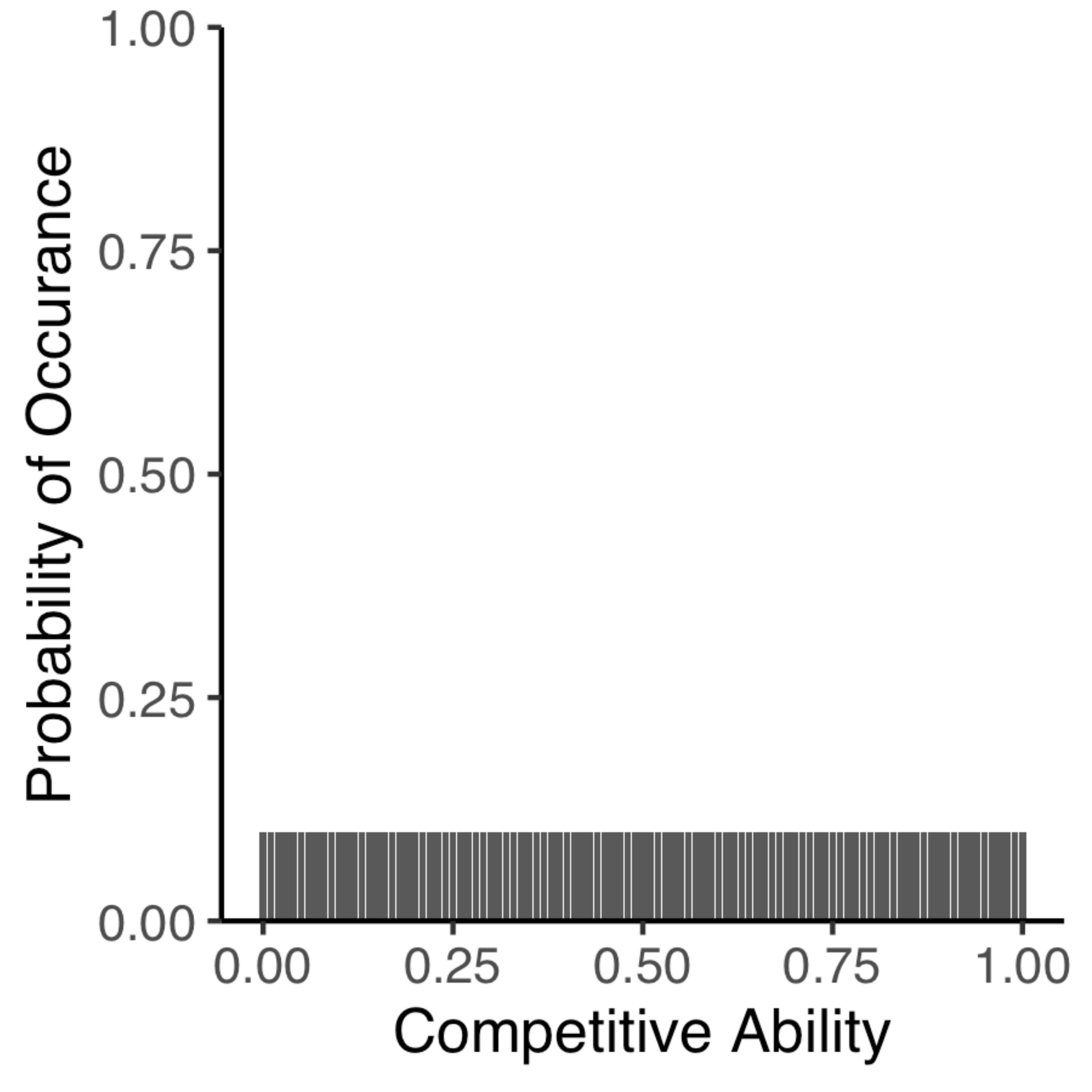}
   \end{subfigure}
   \begin{subfigure}{0.32\linewidth} \centering
     \includegraphics[width=\textwidth]{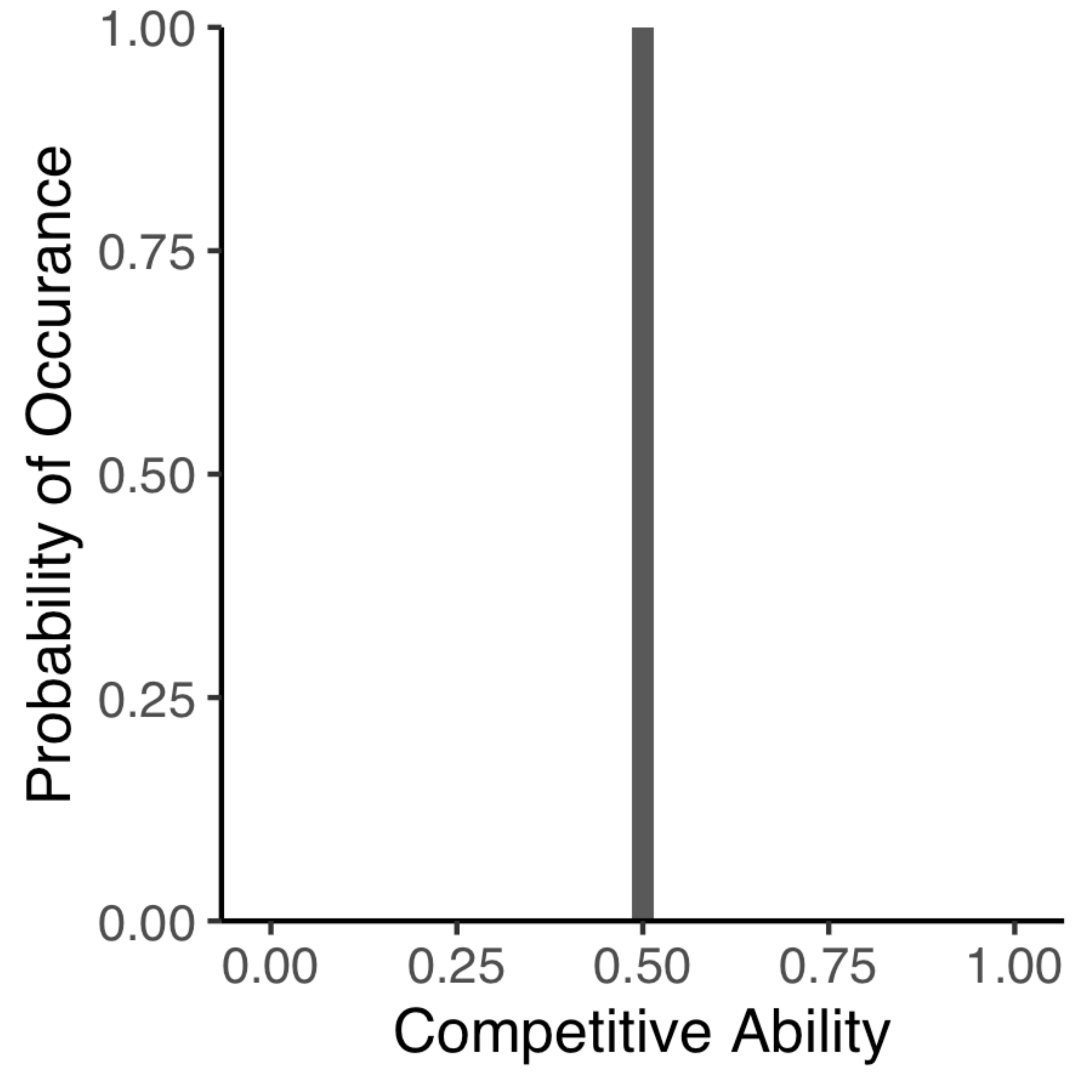}
   \end{subfigure}
     \begin{subfigure}{0.32\linewidth} \centering
     \includegraphics[width=\textwidth]{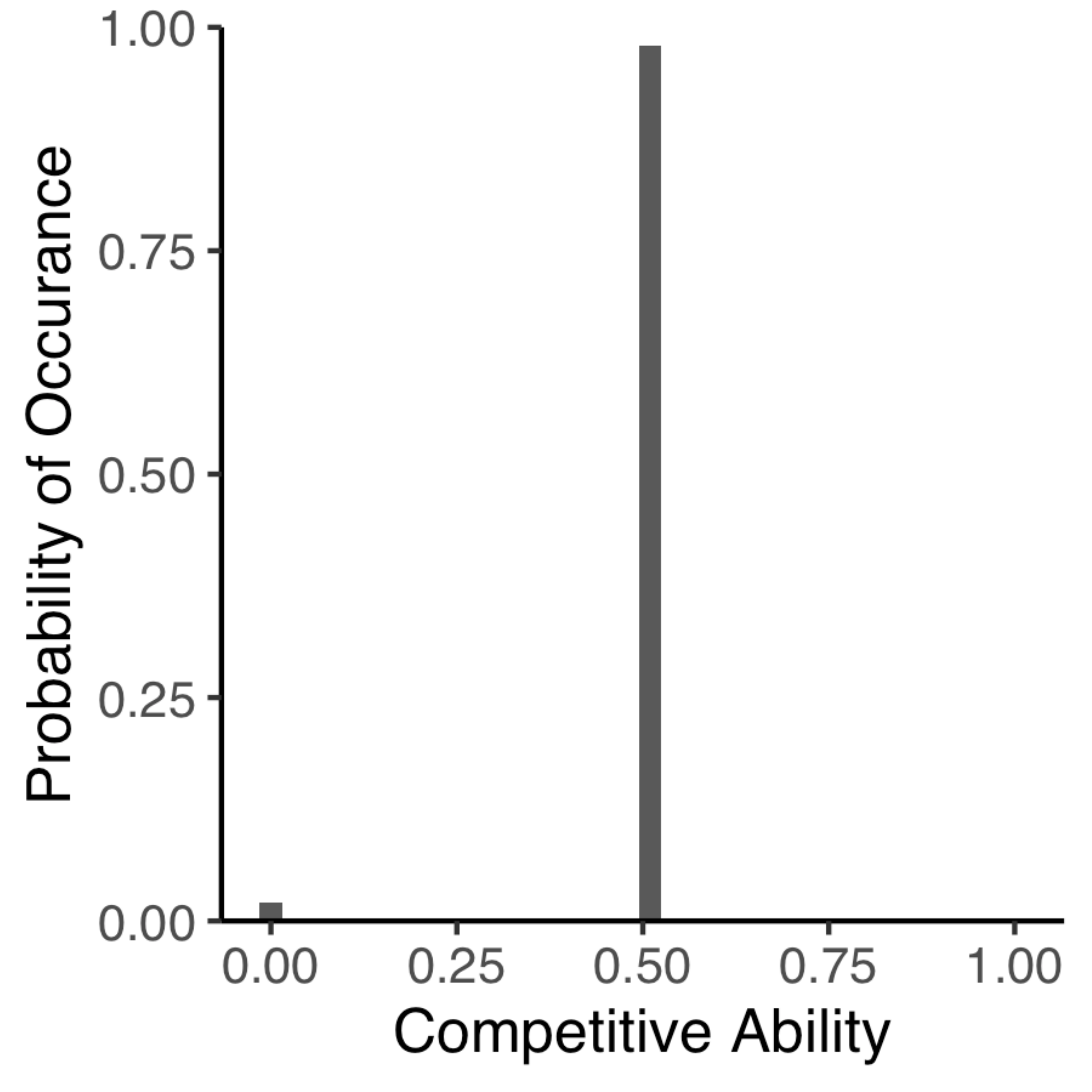}
     \end{subfigure}    
\caption{A uniform distribution has all values between zero and 1 equally likely. An invariant or degenerate distribution is the opposite; all competitive abilities are equal to the mean.  Lake Wobegon is characterized by one zero (loser) and all others have a mean slightly above 0.5, depending on population size.} 
\label{fig:s1}
\end{figure}

To visualize the theoretical mathematics, we simulate competitions for the discrete model, following the model described in \cite{mdr2014}.   Competition outcomes are based on a value for the ‘competitive ability (CA)’ which is a number between 0 and 1, with zero indicating poor competitive ability and 1 indicating an unbeatable winner. The competitive ability reflects fitness with respect to an ecologically meaningful trait of the individual, such as physiology, predator defence, or morphology. The key is that the CA can be distributed in infinite ways and yet yield the same mean competitive ability. Some of the distributions we highlight here are illustrated in Fig \ref{fig:s1}.

\section*{Results and Discussion}
In the following we provide Theorem 1, our main result, that presents all equilibrium strategies for any number of competing species.  We explain how the mathematical proof is obtained by building on the case of two competing species and outline the key results we prove in that special case, but it is important to keep in mind that Theorem 1 applies to \em any number \em of competing species.  We then document the biological meaning through simulations that visualize these results for a subset of cases; the mathematical theorem holds for infinitely many cases so it is of course only possible to simulate and visualize a subset of these.  

\begin{theorem} \label{thm:1} The equilibrium strategies in the discrete case are comprised of precisely those strategies $A$ as in Definition \ref{def:strategy_d} that satisfy $\wp(A; B) \geq  0$, 
for all strategies $B$ as in Definition \ref{def:strategy_d}.  Moreover, $\mca(A) = \frac 1 2$, and $\wp(A; B)=0$ if and only if $\mca(B) = \frac 1 2$.  
In the discrete case, when the set of competitive abilities is $ \left\{ x_j = \frac j M \right\}_{j=0} ^M$, 
and $M$ is odd, then $A(x_j)$ is positive and identical for all $j$.  In case $M$ is even, then $\mca(A) = \frac 1 2$ and $A(x_{2j}) = A(x_0)$, and $A(x_{2j+1}) = A(x_1)$ for $j=0, \ldots, \frac M 2$.  The equilibrium strategies in the continuous case are comprised of non-negative, continuous functions $f$ that satisfy $\wp(f; g) \geq 0$, 
for all strategies $g$ as in Definition \ref{def:strategy_c}.  Moreover, $\mca(f) = \frac 1 2$, and $\wp(f;g)=0$ if and only if $\mca(g) = \frac 1 2$.  In the continuous case, the functions that satisfy this condition are all positive constant functions.  
\end{theorem} 

The following corollary to our main theorem shows that combining species characterized by equilibrium strategies, the resulting combined species are still equilibrium strategies.  
\begin{cor} \label{cor:adding} The sum of equilibrium strategies is an equilibrium strategy in the following sense:  assume that $(A_1, \ldots, A_n)$ is an equilibrium strategy and that $(B_1, \ldots, B_m)$ is an equilibrium strategy, for some $1 \leq m \leq n$.  Then $(A_1 + B_1, \ldots, A_m + B_m, A_{m+1}, \ldots A_n)$ is an equilibrium strategy.  This is true in both the discrete and continuous models. 
\end{cor} 

Since the proof is quite short, we include it here.  
\begin{proof}
We note that the sum of any two strategies contained in an equilibrium strategy has precisely the same features necessary and sufficient to be a strategy comprising an equilibrium strategy. 
\end{proof} 

The techniques we use to prove Theorem 1 are quite different in the discrete and continuous cases, and so we distinguish these cases in \S \ref{si:discrete} and \S \ref{si:continuous}.  In both cases, however, we use the same method of building from the special case of two competing species, and use this case to establish our theorem for arbitrary and possibly fluctuating numbers of competing species.  This technique bears some resemblance to a proof by induction, viewing the `base case' as the case of two competing species.

\subsection*{Results obtained for the discrete model for the special case of two competing species}
A uniform strategy is impervious to invasion by any other strategy.  
\begin{prop} \label{prop0}  A \em uniform \em strategy is defined for $|U|>0$ to satisfy $U(x_j) = \frac{|U|}{M+1}$ for all $j$.  Then for any strategy $A$, with $|A|>0$ and $\mca(A) \leq \frac 1 2$, we have 
\[ \wp(A; U) \leq 0, \quad \wp(U; A) \geq 0,\] 
with equality if and only if $\mca(A) = \frac 1 2$.
\end{prop} 

The following result is not necessarily of independent interest, but it is crucial to proving that asymmetric strategies are vulnerable to invasion.
\begin{prop} \label{prop1} Assume that a strategy $A$ in the discrete case has $\mca(A) = 0.5$ and is not symmetric about $0.5$, then there exists at least one $\ell$ with $0 \leq \ell \leq M$ such that 
\[ A(x_\ell) + 2 \sum_{j=0} ^\ell A(x_j) > A(x_{M-\ell}) + 2 \sum_{j=M-\ell + 1} ^M A(x_j).\] 
\end{prop} 

Using the preceding proposition, we prove that asymmetric strategies are vulnerable to invasion.  It is worthy to note that we explicitly construct the invading strategy in the proof of this result.  
\begin{prop} \label{prop2} Let $A$ be a strategy that has $\mca(A) = \frac 1 2$ that is not symmetric with respect to $\frac 1 2$.  Then there exists a strategy $B$ that has $\mca(B) = \frac 1 2$ for which 
\[ \wp(A; B) < 0, \quad \wp(B; A) > 0.\] 
\end{prop} 

We use all of the preceding propositions to deduce the equilibrium strategies in the special case of two competing species.  
\begin{prop} \label{prop3}  Assume that $(A, B)$ is an equilibrium strategy.  Then 
\[ \wp(A; B) = \wp(B; A) = 0. \] 
Moreover, we have for any strategy $C$, 
\begin{equation} \label{eq:eq_strategy_2} \wp(*; C) \geq 0, \quad * = A, B. \end{equation} 
In case $M$ is odd, all equilibrium strategies are uniform.  In case $M$ is even, all equilibrium strategies have $\mca$ equal to $\frac 1 2$ and further satisfy $A(x_{2j}) = A(x_0)$, $A(x_{2j+1}) = A(x_1)$ for all $j=0, 1, \ldots, \frac M 2$.  Furthermore equality holds in  \eqref{eq:eq_strategy_2} if and only if $\mca(C) = \frac 1 2$. 
\end{prop} 

In the proof of the preceding proposition, given a non-equilibrium strategy, we give a recipe for the construction of a species that can successfully invade and eliminate the non-equilibrium strategy.  Building upon the case of two competing species, we prove Theorem 1 in the discrete case in a somewhat inductive style in \S \ref{si:continuous_t1}.

\subsection*{Results obtained for the continuous model for the special case of two competing species} 
Similar to the discrete case, we analyze the continuous model by starting with the special case of two competing species. 
\begin{prop} \label{prop1:cont}  In the continuous model, a pair of strategies $(f_1, f_2)$ is an equilibrium strategy if and only if both $f_1$ and $f_2$ are positive constants.  Moreover, they satisfy 
\[ \wp(f_i; g) \geq 0, \quad i = 1,2\] 
for all strategies $g$ with equality if and only if $\mca(g) = \frac 1 2$. 
\end{prop} 
In the proof of the preceding proposition, given a non-equilibrium strategy, we give a recipe for the construction of a species that can successfully invade and eliminate the non-equilibrium strategy.  Building upon the case of two competing species, we prove Theorem 1 in the continuous  case  in \S \ref{si:continuous_t1}.

\subsection*{A comparison of equilibrium strategies and evolutionary stable strategies} 
Here we have obtained a rigorous mathematical proof that the best strategies for competing species have the maximum possible mean competitive ability and will defeat all strategies that have lower mean competitive ability, as one would expect.  Surprisingly, these strategies are 
 \em neutral \em towards all other species that also have a strategy with the maximal mean competitive ability.  An immediate consequence is that there are no evolutionary stable strategies; this is consistent with the results of \cite{mdr2018}.

The distinction between evolutionary stable strategies (ESS) and equilibrium strategies is that whereas an ESS is the best strategy to face a specific challenge, equilibrium strategies should be understood as the best strategies to simultaneously face all challenges.  For example, a specific ecological challenge faced by many photosynthetic microbes is balancing the tradeoff for light and nutrients that are available in opposite concentration gradients in surface waters of lakes and oceans. Although the optimal depth that best balances the tradeoff for light and nutrient resources may be an evolutionary stable strategy \cite{klausmeierlitchman}, that does not apply when conditions change, such as when predators appear.  Algal aggregations within a single depth or a patch have been shown to be necessary for survival of zooplankton predators \cite{mullinbrooks}, and empirical research has supported this prediction \cite{mendendeuerfredrickson}.  In general, plankton inhabit a patchy environment and contribute to environmental heterogeneity \cite{durhamstocker}.  So, while an ESS for a particular scenario of constraints may be found, such as for resource acquisition, it would not apply to other ecological challenges like predator avoidance. 
 
Equilibrium strategies are characterized by maximal phenotypic heterogeneity in the sense that the competitive abilities are spread over the entire range of possible values rather than clustered around the mean.  For any strategy that is not an equilibrium strategy, in \S \ref{si:discrete} for the discrete model and \S \ref{si:continuous} for the continuous model we design a competitor strategy that is cumulatively equally fit, yet due to the particular features of its trait distribution can outcompete the non-equilibrium strategy, leading to the extinction of the less variable competitor.  For example, in the discrete model an invariant or degenerate strategy, lacking variability, has a probability of 0 for every competitive ability not equal to the mean.  One of many strategies that will defeat the invariant strategy is Lake Wobegon.\footnote{We named this ‘designer distribution’ Lake Wobegon based on a  North American radio show (Prairie Home Companion) because the feature of that distribution is that although the mean competitive ability is 0.5 (as for all) but for this strategy, one individual has a competitive ability of 0 which results in all others having a slightly above average competitive ability, depending on population size. As the saying in the show went, ‘all the children {except one} are above average.”}   Such a strategy assigns the majority of individuals' competitive ability slightly above the mean competitive ability, while one individual is assigned competitive ability 0, in this way guaranteeing that the mean competitive ability of the Lake Wobegon strategy is equal to that of the invariant strategy. Similarly, we give a recipe for constructing a superior strategy to outcompete any non-equilibrium strategy in \S \ref{si:discrete} for the discrete model and in \S \ref{si:continuous} for the continuous model.  We do not claim that these designer strategies reflect biological realism.  Instead, they demonstrate that the equilibrium strategies are superior in their ability to coexist with any other strategy and resist replacement through competition or invasion. 

Considering the heterogeneous and dynamic environment inhabited especially by marine microbes, the lack of an evolutionary stable strategy \cite{mdr2018} is logical and drives the need for a more general concept.  Accommodating the specific characteristics of microbial populations, we analyzed the fitness of species by interpolating between individual-level competitions and species-level population growth or decline.  Previous work has conceptualized species coexistence in microbial communities by invoking spatial or temporal heterogeneity in some form \cite{kerr2002, dolson2017, hart2017}. Consequently, when placed in a common environment, species must have different mean fitness, in the sense that they must be locally adapted to different environmental niches. In contrast, our theory shows that multiple species with the same average fitness in the same environment (niche) can coexist indefinitely, as long as individuals within those species vary maximally in their competitive abilities.   The theory predicts survival of the most variable and equally cumulatively fit individuals and thus predicts co-existence of such species.  This finding implies that selection in generally adapted species favors maintenance of variability rather than the evolution of optimized trait values but also that the production of intra-specific or intra-genotype variation in trait values is adaptive. The genetic maintenance of intra-genomic variability has been reported \cite{bruijning2020}. The distribution of risk afforded by phenotypic heterogeneity also points to a potential adaptive reservoir in physiological, behavioral and genetic diversity in light of large-scale changes occurring in ecosystems, including climate change \cite{moran2016, kaye2017}.  

\subsection*{Intra-specific phenotypic variability thwarts invasion, promotes co-existence, and allows emergence of new species.} 
Having identified the key role of phenotypic heterogeneity in species competition and coexistence, we explored the ramifications of specific distributions of competitive strategies commonly used in ecological models, such as gaussian, bimodal, and degenerate distributions.  These types of distributions are often assumed to underlie the metric under study \cite{gotelli2008}. To visualize the effects of specific strategies representing different degrees of phenotypic heterogeneity on the outcomes of competition, persistence and invasion, we utilized an individual-based competition simulation that emulates the rules of our game theoretic approach \cite{mdr2014}. Groups of individuals, which can represent populations or species, compete in a randomized design as depicted in Fig \ref{fig:2}.  These competitions are based on the discrete model. 
We assume no heritability of a specific competitive ability across generations and rather invoke the maintenance of the strategy at the species level. This assumption is well supported by empirical observations that individuals modulate key characteristics, including variations in growth rate, as a function of the presence of conspecifics or competing species \cite{wolf2017, collinsschaum}. This shows that rather than traits that are strictly executed based on environmental conditions, traits are further modulated by which species or individuals are present. The outcomes of these competitions on population abundance are evaluated in an individual based model as shown in Figs \ref{fig:3a}, \ref{fig:3b}, and \ref{fig:3c}. In each instance, the simulations reproduce the exact outcomes predicted by the mathematically-derived payoff functions.  These simulations provide a visualization of the game theoretic model for the discrete case.

\begin{figure}
 \centering
  \begin{minipage}{\textwidth}
    \includegraphics[width=\textwidth]{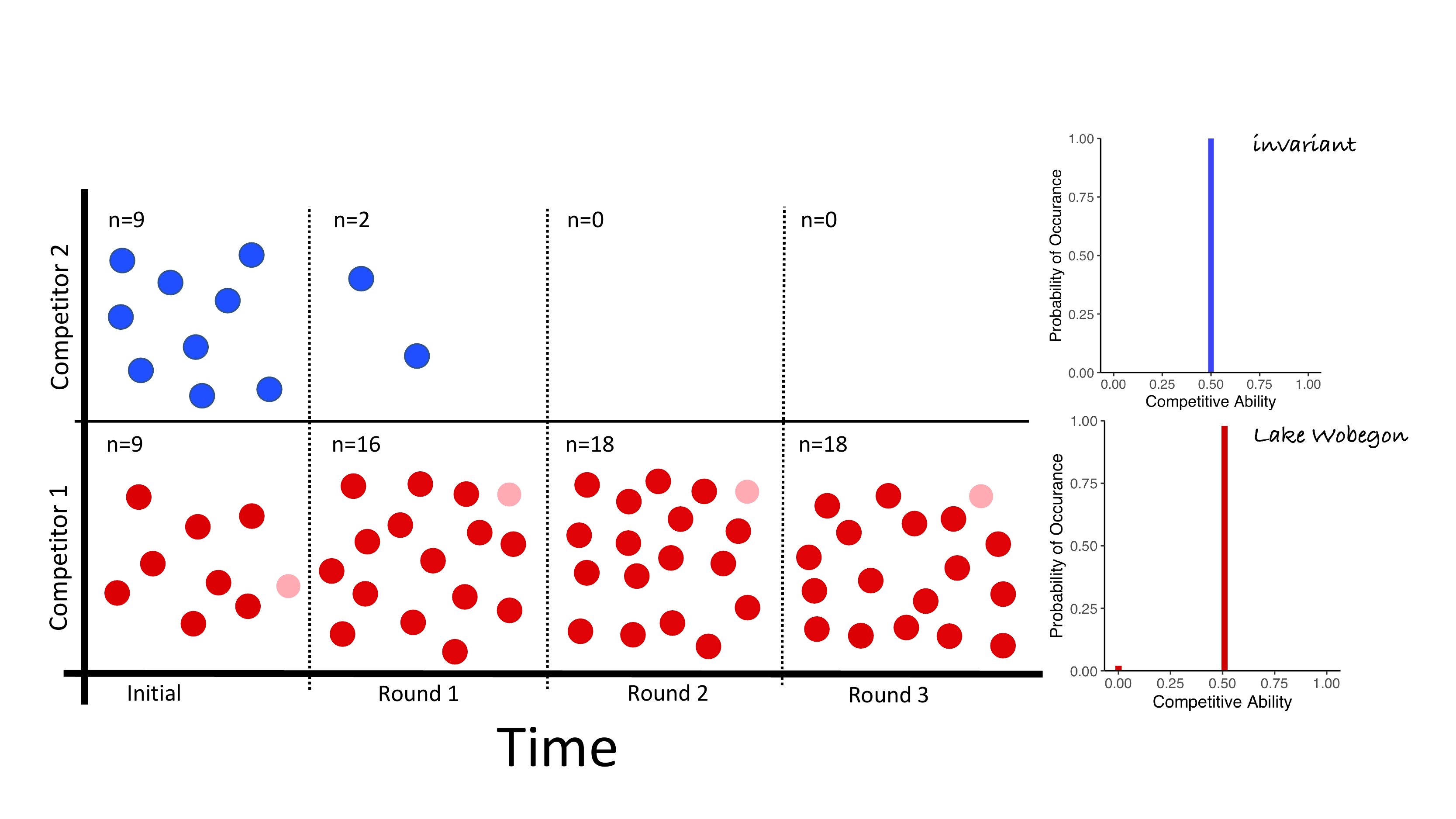}
    \caption{Three rounds of competition are simulated for two strategies.  Probability distributions of strategies are illustrated on the right with individuals (n) represented by color coded circles. Here the invariant strategy (top, blue), characterized by lack of variability and a constant competitive ability equal to the mean, is eliminated by the Lake Wobegon strategy (bottom, red), characterized by one individual with 0 competitive ability and all others are ‘above average’ (mean+x for a small $x>0$).}
        \label{fig:3a}
  \end{minipage}
\end{figure}

\begin{figure}
  \centering
  \begin{minipage}[b]{\textwidth}
    \includegraphics[width=\textwidth]{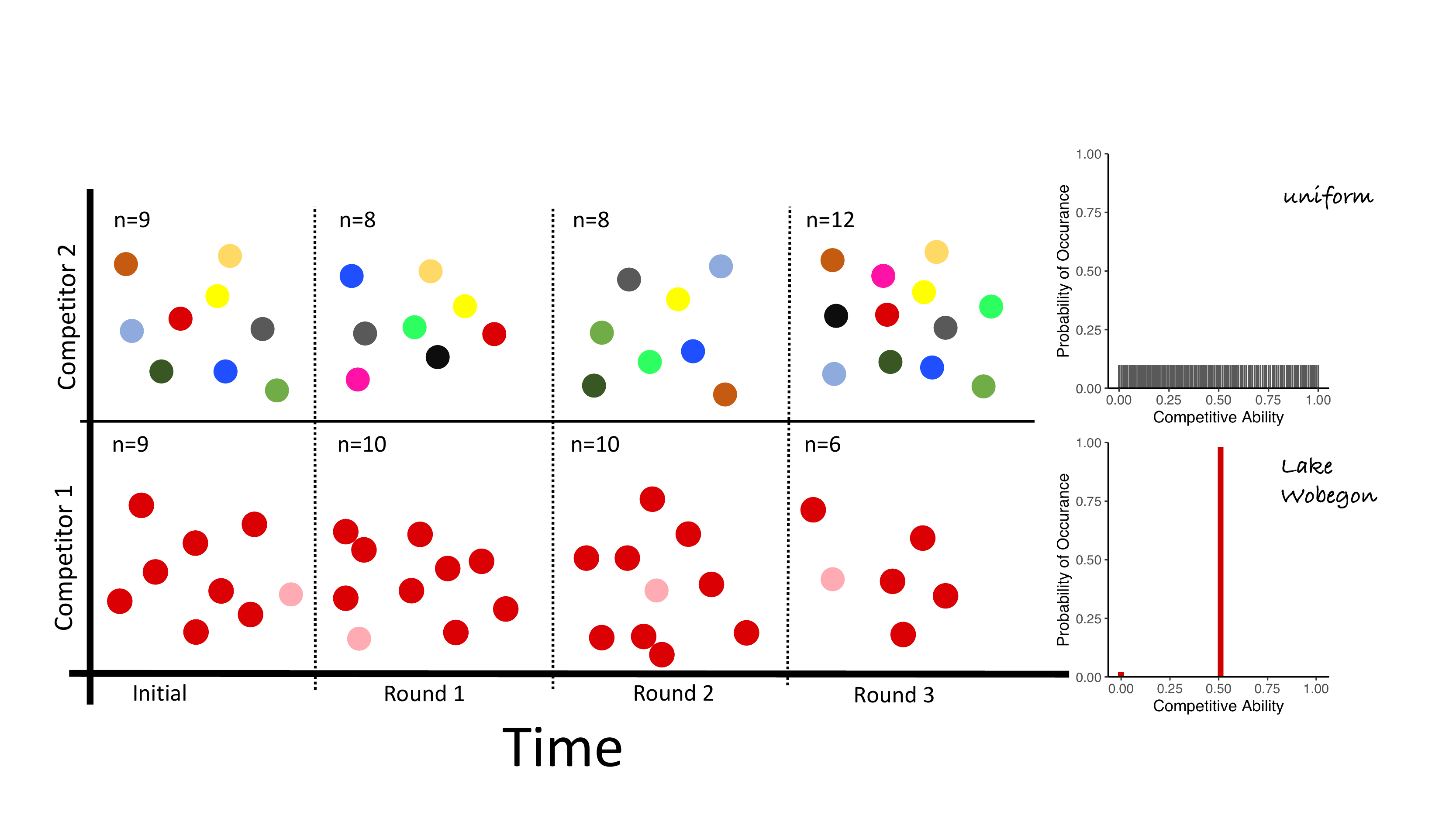}
    \caption{In the same type of three time-step competition as in Fig \ref{fig:3a}, the Lake Wobegon strategy coexists with an equilibrium strategy (top, multicolored).  Note the color range for the equilibrium strategy reflects its maximal phenotypic variability.}
    \label{fig:3b} 
  \end{minipage}
 \end{figure} 
 \begin{figure}[h] \centering
  \begin{minipage}[b]{\textwidth}
    \includegraphics[width=\textwidth]{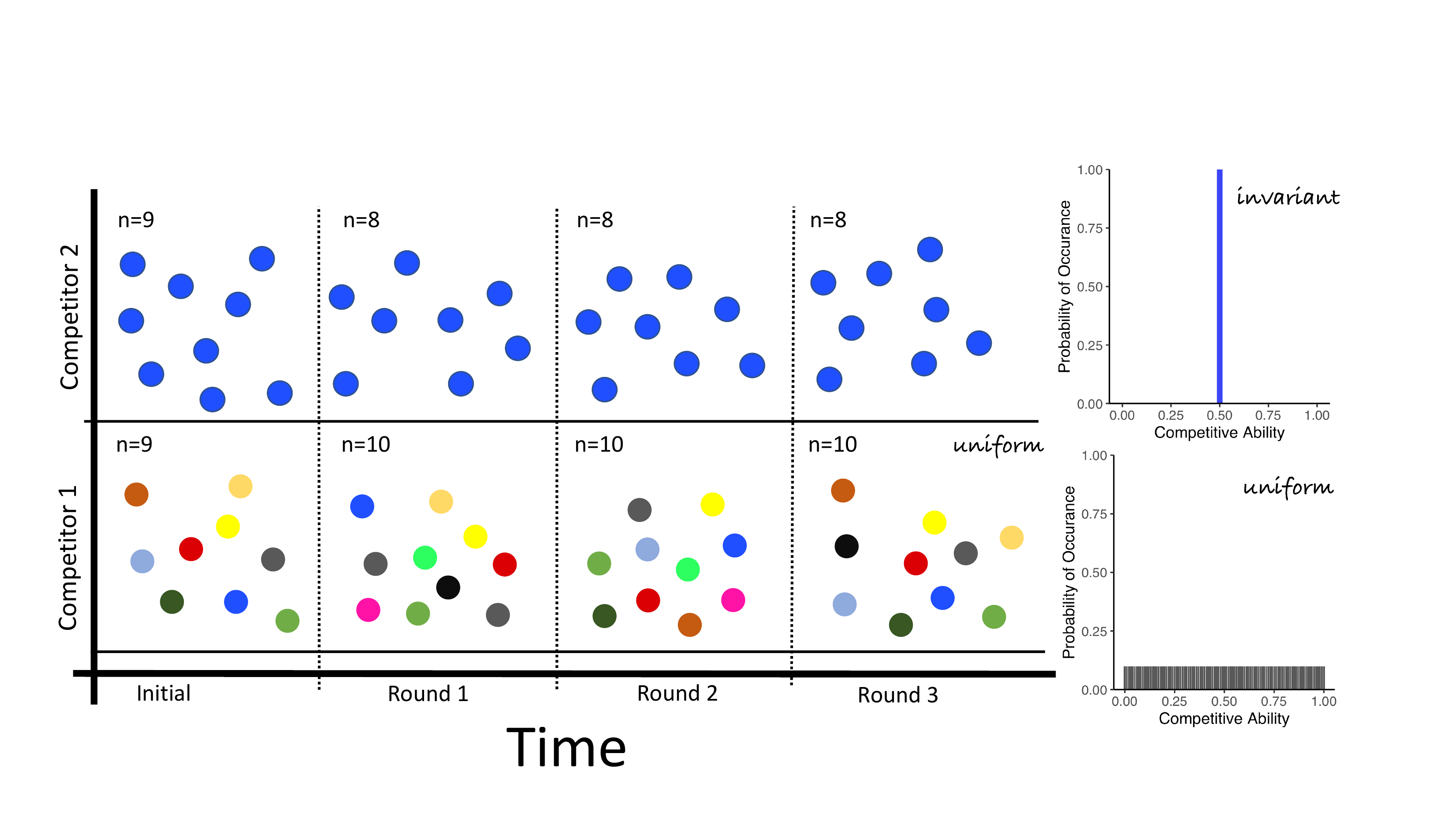}
    \caption{An equilibrium strategy also co-exists with the invariant strategy.}
    \label{fig:3c}
  \end{minipage}
\end{figure}

Most notably, our theory predicts that equilibrium strategies are capable of co-existing with any other strategy that has the same mean competitive ability and cannot be replaced by another strategy that has the equal or lower mean competitive ability. The predicted co-existence of the maximally  variable, uniform distribution is illustrated for some competitor strategies in Fig \ref{fig:4A}.  The theory also predicts that a population characterized by an equilibrium strategy starting orders of magnitude lower in abundance than its competitors can persist, while an invariant population at high abundances can be replaced.  Again the simulations display the theoretical predictions as shown in Fig \ref{fig:4BC}.  The robustness of the equilibrium strategies starkly contrasts the weakness of the invariant strategy, where a population is characterized only by its average competitive ability with no variability about the mean.  Remarkably, several strategies can be identified that can replace the invariant strategy as shown in Fig \ref{fig:4BC}, suggesting that the invariant strategy should not persist when other populations of equal mean fitness are present. Thus, a common practice of representing empirical measurements of, for example, microbial physiology as averages may mask the important information of variance in physiological capacity. Population size, with exception of very small population size as discussed in \S \ref{ss:extinction}, and duration does not affect these results.  Naturally, there are cases where extinctions are observed as shown in \S \ref{ss:extinction}. Our simulations faithfully reproduce the predictions of the competitive exclusion principle \cite{hardin1960} when competition involves an overall inferior competitor, characterized by a lower mean competitive ability; see \S \ref{ss:extinction}. Extinctions of equilibrium strategies are observed in cases when the population size is very low (10s of individuals) because equilibrium strategy distributions cannot faithfully be reproduced when few individuals represent it; see \S \ref{s:simulations}. This is consistent with chance, rather than selection, dominating persistence for very small populations.  Aside from these explainable deviations, the finding that a maximally-variable distribution of ecological traits will persist in competition has remarkable implications for the ecology and evolution of free living microbes and how we should study them. Maximal intra-specific variability supports unlimited co-existence of species with equal mean fitness, as demonstrated in a simulation with 100s of populations, all represented by maximally intra-specific strategies as shown in Fig \ref{fig:5}.

\begin{figure}[h]
\centering
   \begin{subfigure}{0.49\linewidth} \centering
     \includegraphics[width=\textwidth]{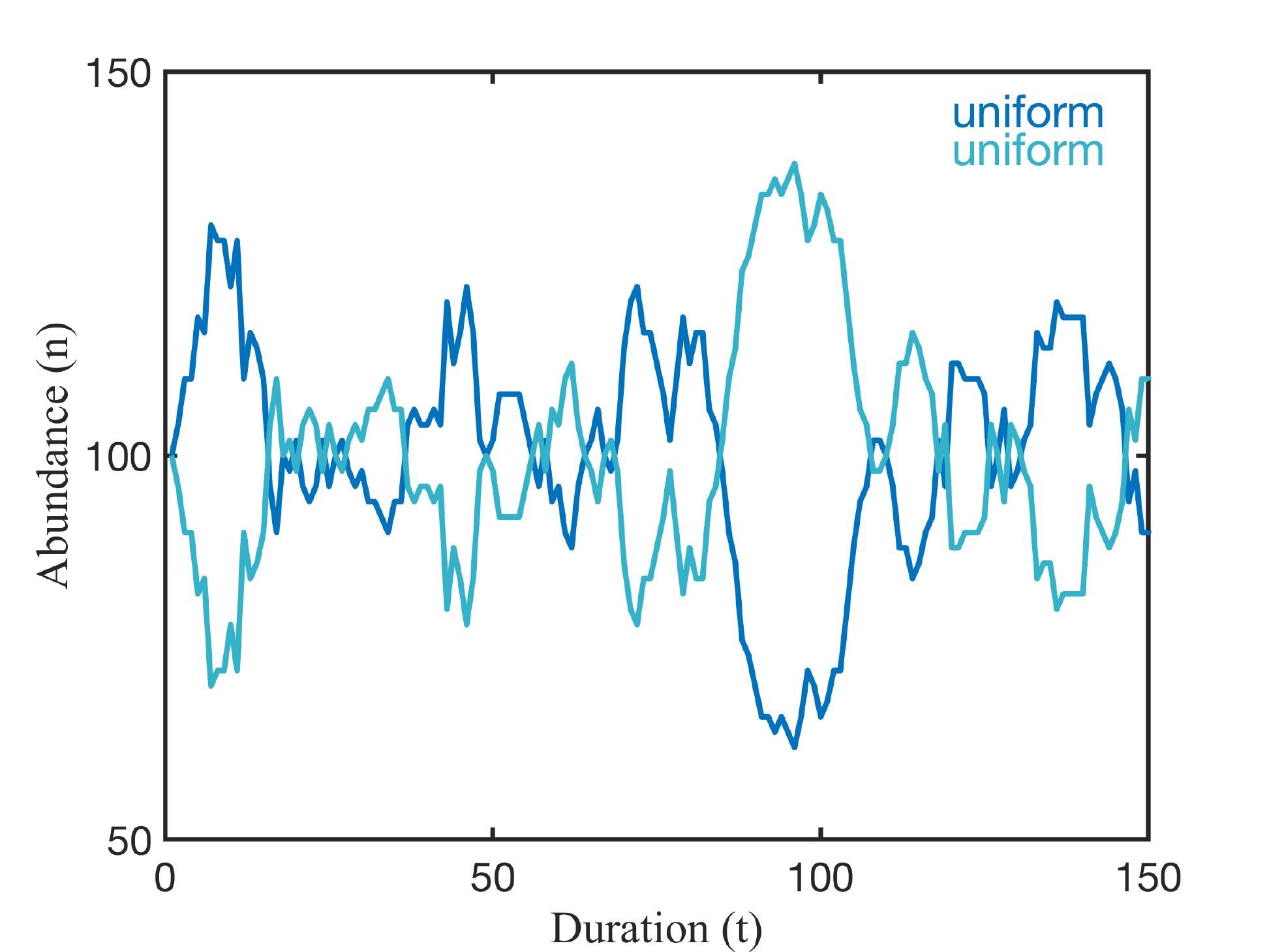}
   \end{subfigure}
   \begin{subfigure}{0.49\linewidth} \centering
     \includegraphics[width=\textwidth]{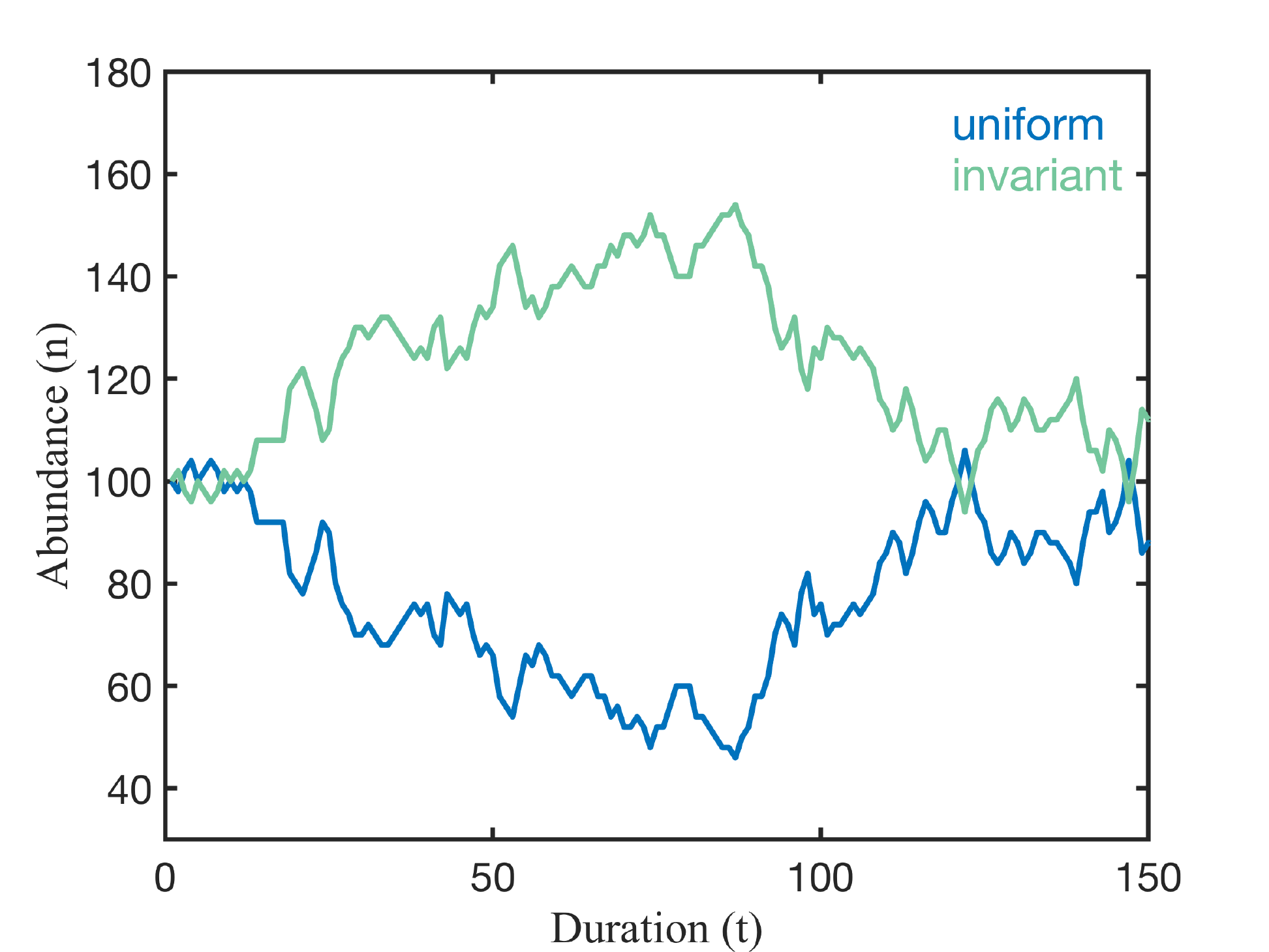}
   \end{subfigure}
     \begin{subfigure}{0.5\linewidth} \centering
    \includegraphics[width=\textwidth]{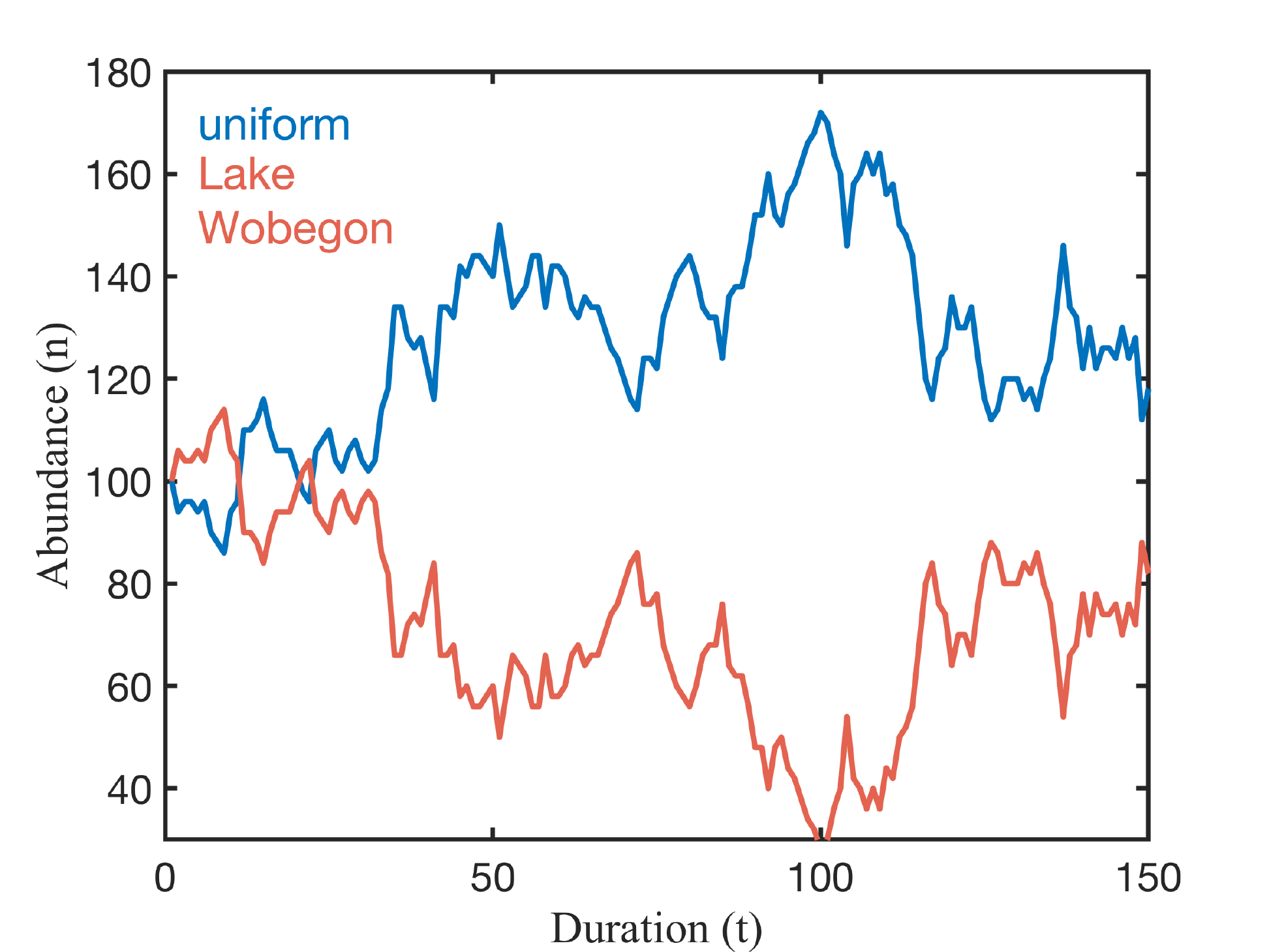}
   \end{subfigure}
\caption{A uniform distribution is an equilibrium strategy.  It co-exists with all strategies that have the same mean competitive ability.} \label{fig:4A}
\end{figure}

\begin{figure}
\centering
   \begin{subfigure}{0.49\linewidth} \centering
     \includegraphics[width=\linewidth]{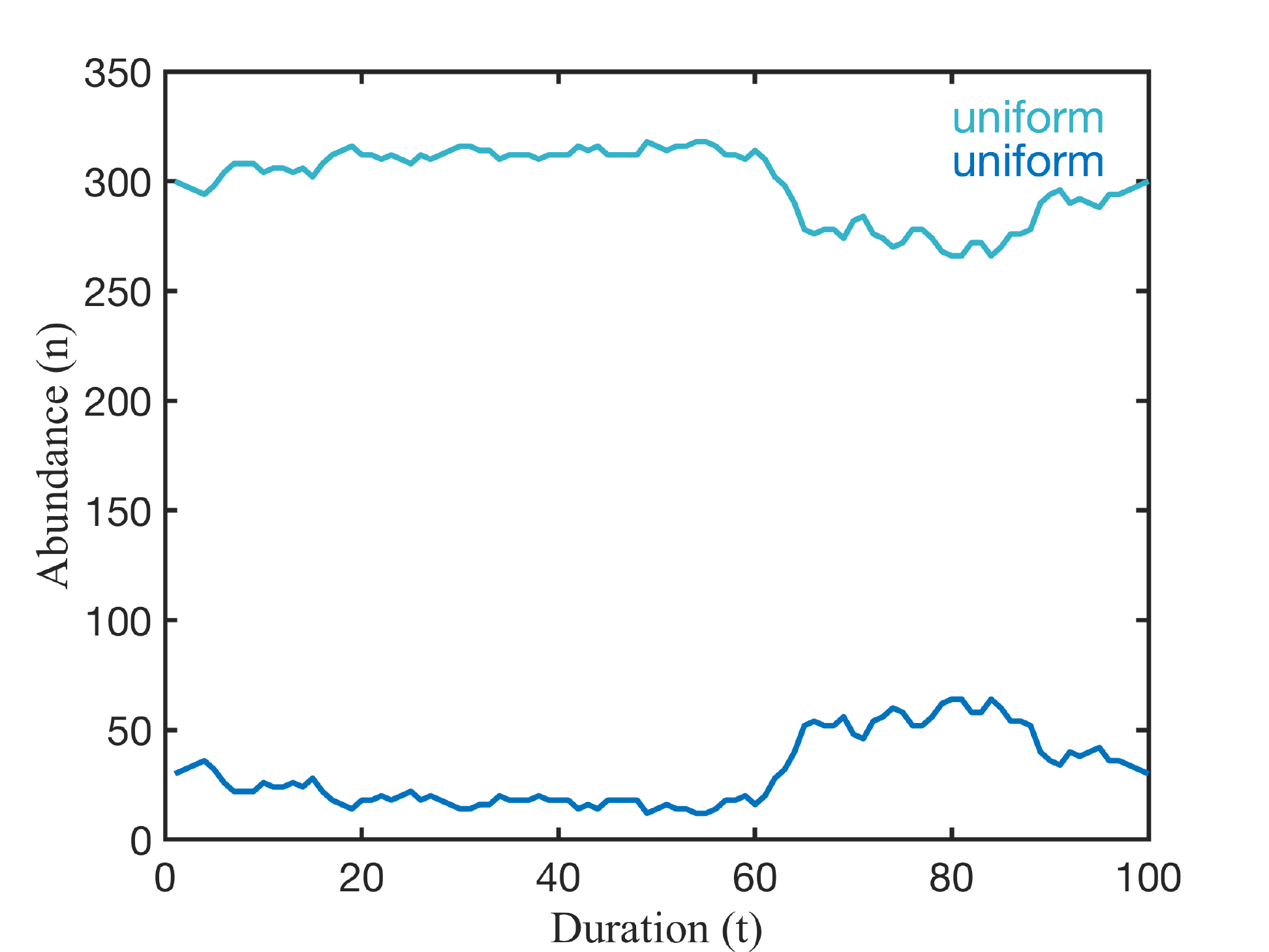}
   \end{subfigure}
   \begin{subfigure}{0.49\linewidth} \centering
     \includegraphics[width=\linewidth]{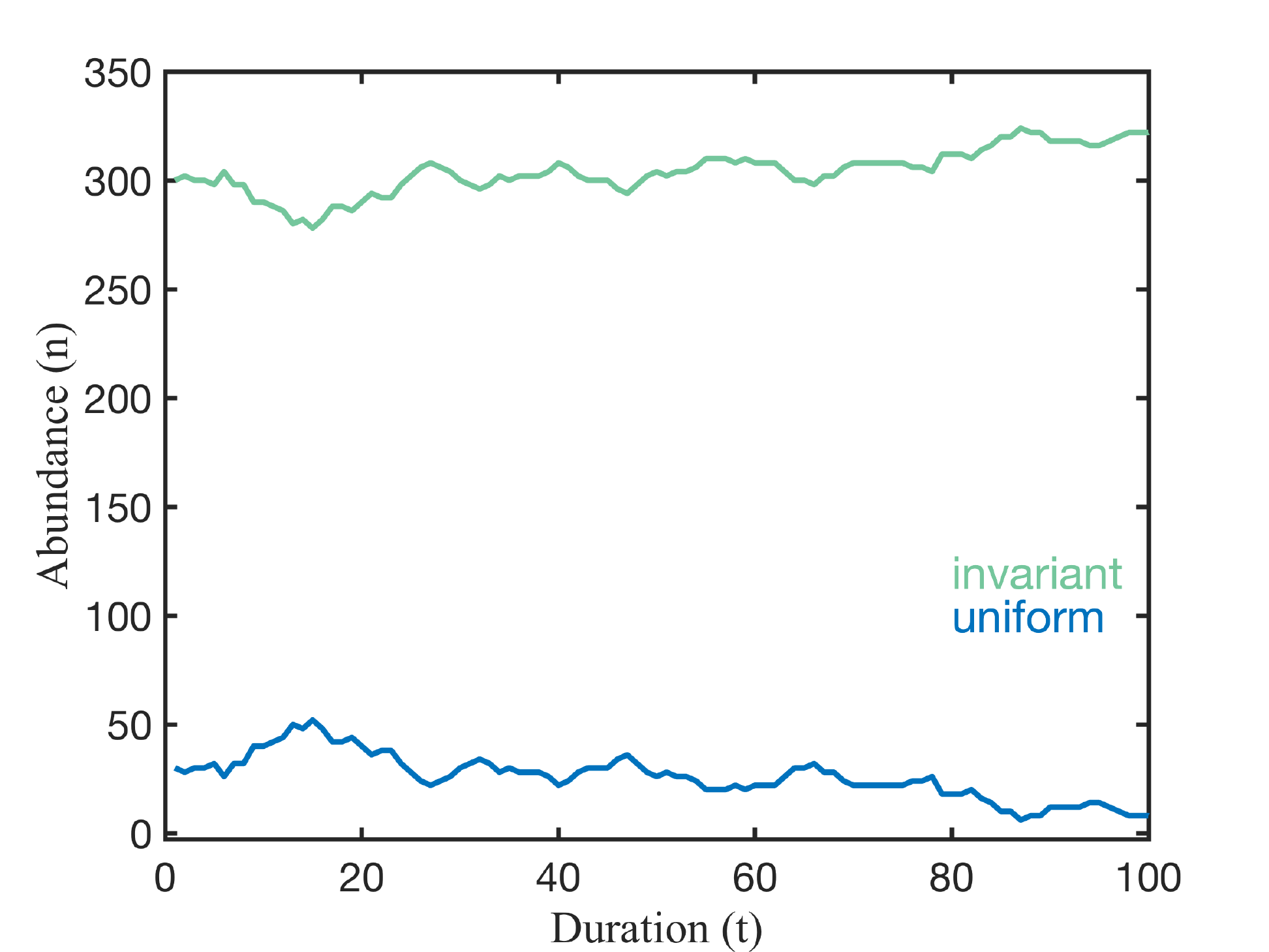}
   \end{subfigure}
     \begin{subfigure}{0.49\linewidth} \centering
     \includegraphics[width=\linewidth]{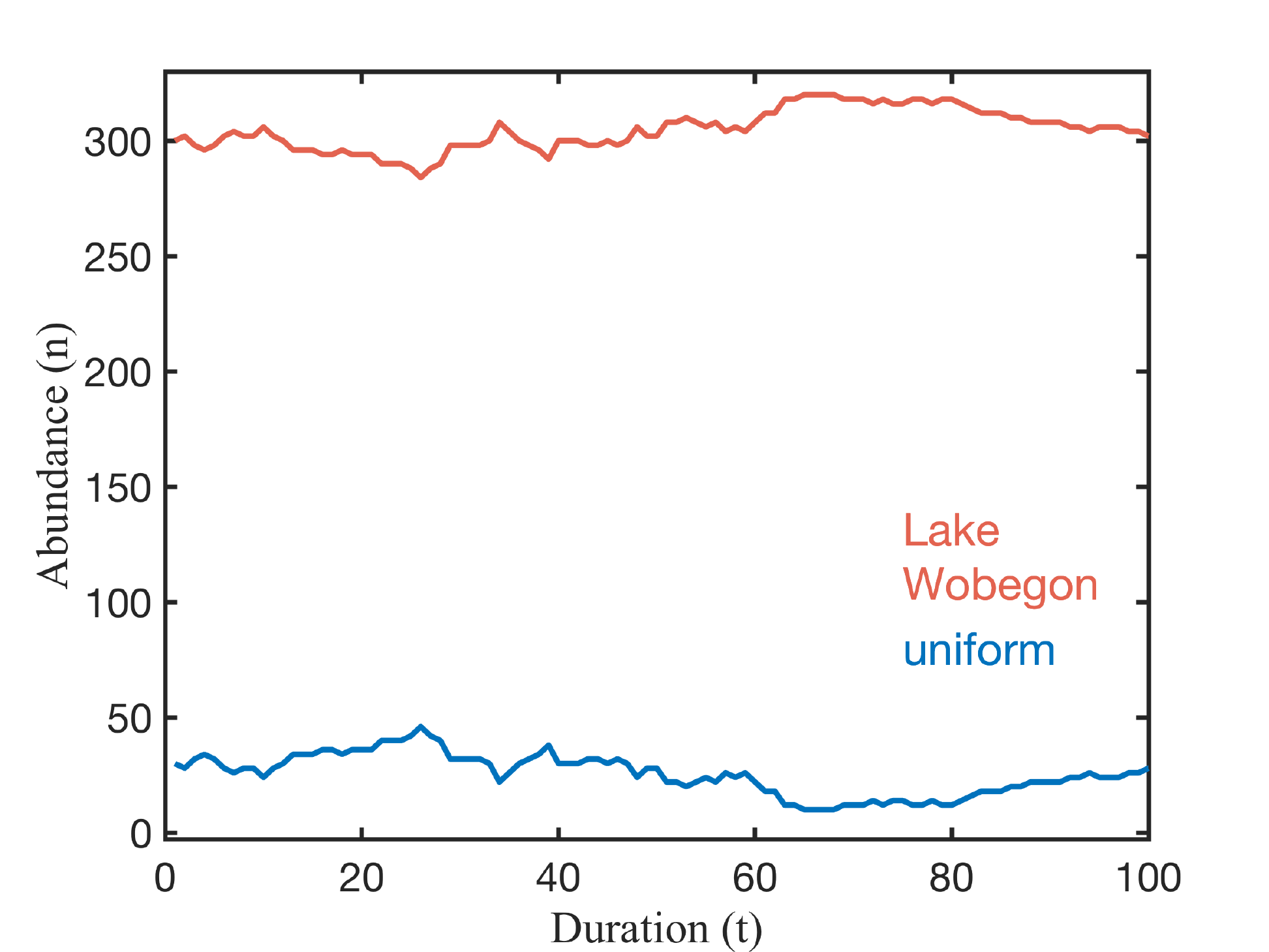}
     \end{subfigure} 
     \begin{subfigure}{0.49\linewidth} \centering
     \includegraphics[width=\linewidth]{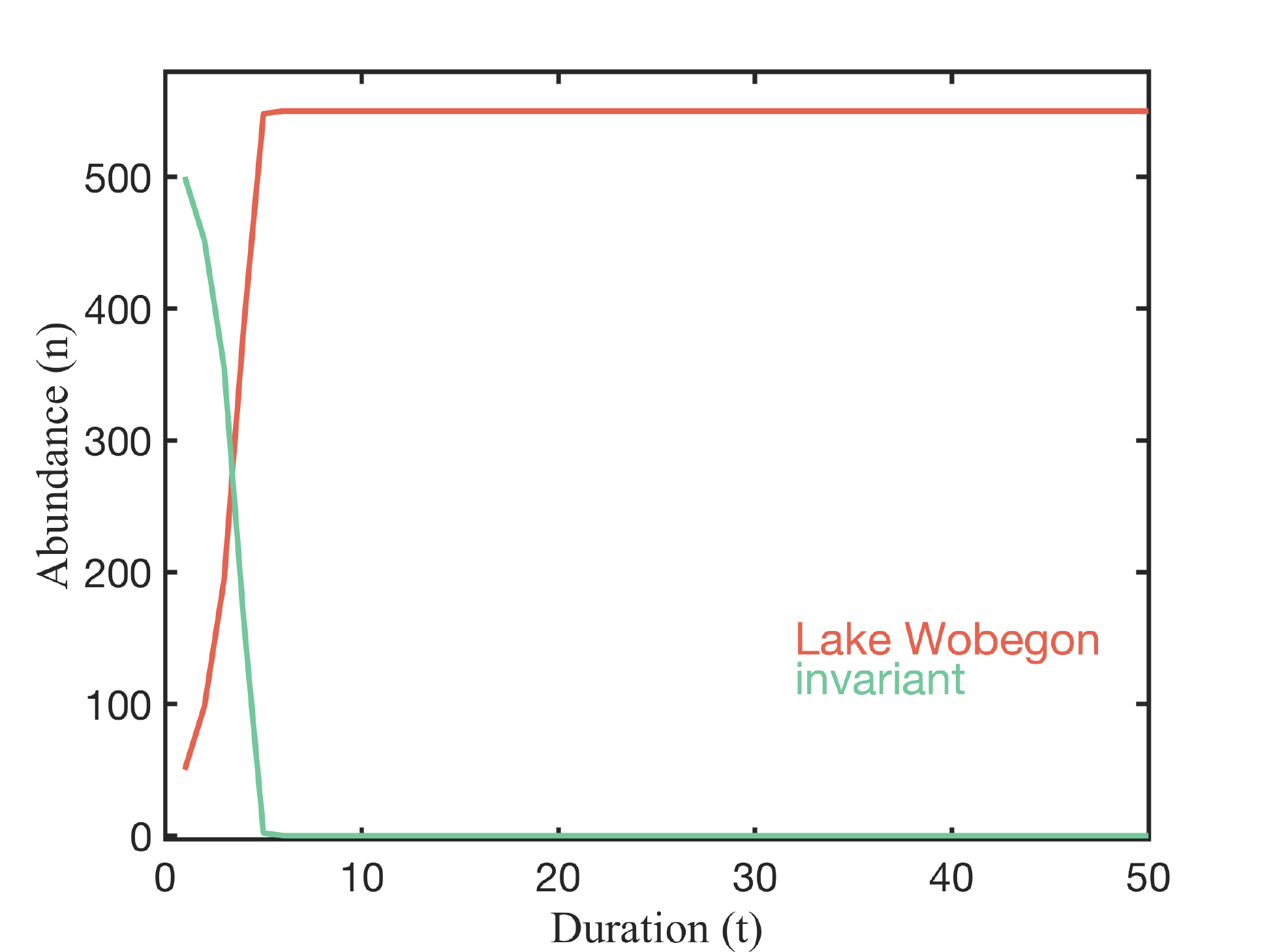}
   \end{subfigure}
\caption{The uniform distribution, an equilibrium strategy, persists even when their population is vastly lower than the competitor's.  In contrast, the invariant strategy starting at a much higher population is quickly replaced by the Lake Wobegon strategy.  Note that the mean competitive abilities of all strategies are identical.} 
\label{fig:4BC}
\end{figure}

These results provide a fundamentally different mechanism for species co-existence than those previously advanced.  There are several prior demonstrations that multiple species can coexist in the same niche or consuming the same resource, thus circumventing the competitive exclusion principle, such as two predators existing on different life stages of the same prey species, or asynchronous introduction of two competitors into a system \cite{armstrongmcgehee} and references therein, \cite{kerr2002}. These previously advanced solutions are similar in that they introduce some asynchrony in the competing species, thus essentially separating competitors and eliminating direct competition.  This is equivalent to invoking environmental patchiness in either space or time to increase the number of available niches.  This type of asynchrony likely enhances species co-existence in numerous cases but is fundamentally an extension of the competitive exclusion principle; it expands the number of niches but still relies on an assumption of one species per niche.  Our results agree with the two species model in \cite{barbarasdandrea} but contradicts their suggestion that in multi-species competition, intraspecific variability depresses biodiversity. There is no mathematical contradiction because \cite{barbarasdandrea} considered a single trait with variations around a normal distribution whereas we mathematically analyze all traits and all distributions.  The mechanism advanced here directly confronts competitors and is flexible in terms of the assumptions regarding population sizes and other specifics. We show that 100's of species can coexist in the same niche despite direct cell-cell competitions simply by incorporating the high phenotypic heterogeneity that is one of the empirically observed, fundamental characteristics of microbial populations.

This study also identifies a novel adaptive advantage (persistence in the face of competitors) and consequence (maintenance of diversity) of phenotypic heterogeneity. Phenotypic heterogeneity has previously been suggested to be adaptive because it can increase survival in heterogenous or fluctuating environments by allowing genotypes to bet-hedge \cite{ackermann2015}, but there are no studies addressing how phenotypic heterogeneity may be advantageous in a homogeneous environment, and why natural selection might favor a particular shape of phenotypic heterogeneity. We demonstrate that because competitive interactions occur between individuals, the heterogenous condition can just as easily stem from the phenotypic heterogeneity in capabilities of competitors rather than features of the abiotic environment. Our findings explain how very high levels of genotype diversity \cite{rynearsonambrust, kashtan2014} can be maintained without invoking niche partitioning within species. 

 \begin{figure}[h] \centering
  \begin{minipage}[b]{\textwidth}
    \includegraphics[width=\textwidth]{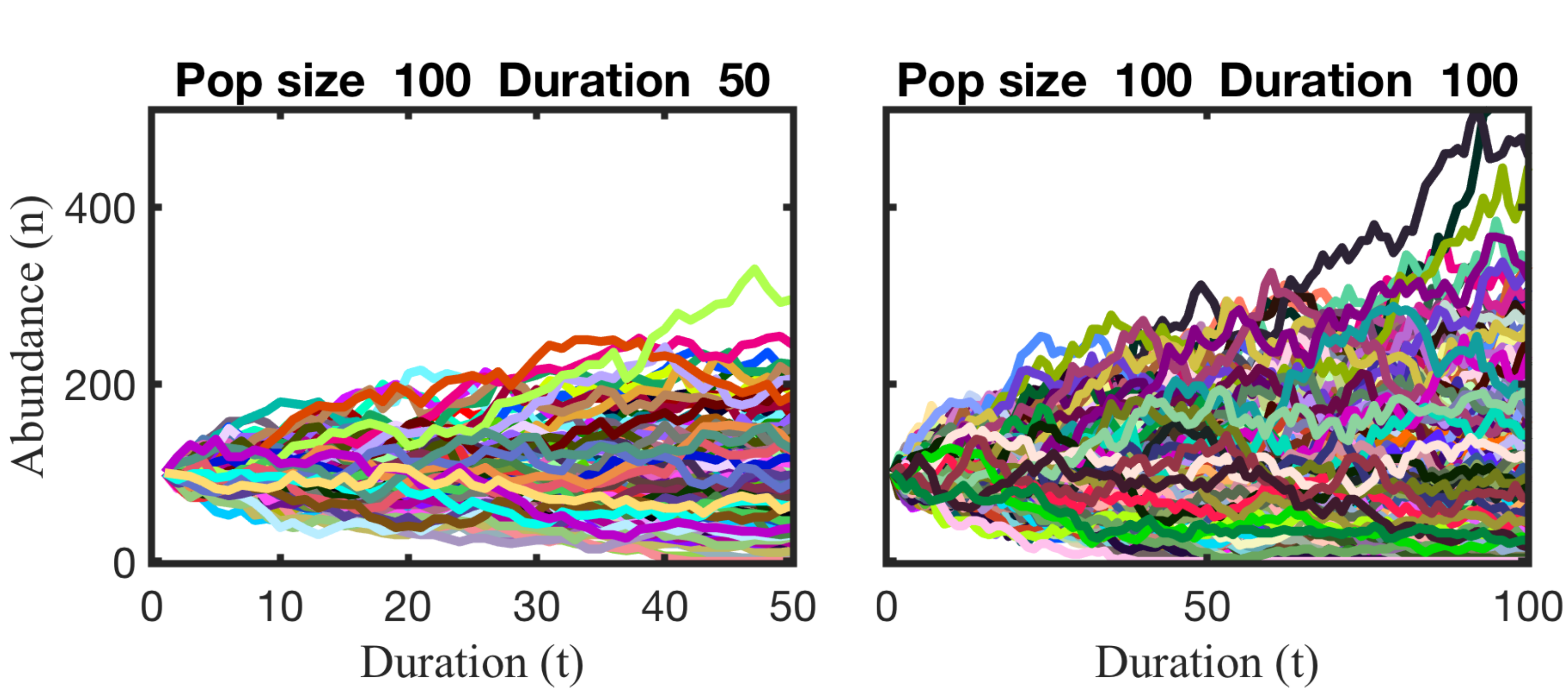}
    \caption{Abundance of 101 (left) and 301 (right) equilibrium strategies in 50 to 100 rounds of competition. Although abundance fluctuates over time and no one population dominates, there are less than a handful of observable extinctions. The number of co-existing species displayed is limited by CPUs and visibility. Theoretically, the number is infinite.}
    \label{fig:5}
  \end{minipage}
\end{figure}

\section*{Conclusion:  A paradigm shift is needed to adequately conceptualize microbial populations.}
Our theory implies that a shift in perspective is needed when conceptualizing microbial populations. A microbe-centric perspective is necessitated by the fact that microbial ecology and evolution are subject to very different principles than the multi-cellular organisms upon which much biological theory was developed. The fundamental importance of microbes requires that the shape of the distribution of competitive strategies be investigated as a possible unifying and structuring principle governing the biodiversity of asexually reproducing microbes. We predict that maximally variable trait distributions will be typical, because of their adaptive features. Our results also suggest that high phenotypic heterogeneity is expected in populations that persist through multiple rounds of competition even in the absence of environmental patchiness and that such heterogeneity may largely reflect inherent characteristics of species \cite{matson2016} rather than methodological limitations or differences \cite{schaum2012}.  Interestingly, the discovered equilibrium strategies display the very same features observed in microbes: (1) these strategies have maximal intra-specific variability and (2) there are an unlimited number of these strategies which are neutral towards each other, implying they cannot universally competitively exclude each other.  The mathematics reflects the biological qualities of microbes.  This fundamental model for microbes can accommodate more complicated variations, which we have not explored here in order to convey the fundaments of the theory.

Disease dynamics, ecosystem function and evolutionary insights can be propelled forward by accounting for the direct linkage from individuals to communities to macroscopic processes such as the abundance and distribution of organisms without the loss of insight imparted by averaging out variability.  We suggest that the inherent phenotypic heterogeneity may be an essential determinant of microbe ecology. Remarkably, the role of phenotypic heterogeneity, is in direct contrast to diversity maintenance mechanisms that rely on ever-increasing numbers of niches because of spatial and temporal heterogeneity \cite{levy2015, hart2017}, or more inclusive lists of traits that underlie fitness \cite{cadier2019}. If true, our theory suggests a fundamentally different mechanism, one where trait selection favors variable distributions, rather than an optimum value with minimal variation. Furthermore, the distribution is at least as important as the average.  There should be stabilizing selection on the generation of phenotypic heterogeneity in an adapted population, rather than a minimization of that variation around an optimum value. 

A major implication of these findings is that entropy in the system always arises. Biological diversity continuously expands. It has not escaped our notice that we may have also discovered a mechanism to explain one of the fundamental laws in biology, that of ever-increasing complexity in evolving systems \cite{mcsheabrandon}.  Our focus here was chiefly on marine microbes because we can support our arguments with empirical observations and constrain descriptions to specifics.  Yet, there is no reason why the same rules should not apply to other, free-living microbes.  Rather than habitat, individual variability is a key factor in species persistence and diversity. Based on the mathematics, there is nothing to suggest that other free-living microbes would not be subject to similar principles. In fact, we believe that here we present a unifying theory that phenotypic heterogeneity enhances biodiversity.

\section*{Acknowledgments}
We thank C.-J. Karlsson and B. Chen for providing constructive criticism on a preliminary draft of this manuscript and H. McNair and P. Marrec for help with the figures.  This work was supported through the National Aeronautics and Space Administration programs: EXport Processes in the global Ocean from RemoTe Sensing (EXPORTS) field campaign (grant 80NSSC17K0716), the Swedish Research Council Grant 2018-03873, and the Australian Research Council Grants ARC DP190103302 and ARC DP190103451. Further support was provided by National Science Foundation awards 1736635 (SMD), 1638834 (TAR) and 1543245 (TAR). National Science Foundation award DMS-1440140 supported JR at the Mathematical Sciences Research Institute in Berkeley, California.

\section*{Supporting information} 
\subsection*{Mathematical proofs for the discrete model} \label{si:discrete} 
We determine the equilibrium strategies by proving several propositions, building up from the case of two competing strategies to arbitrarily many competing strategies.  

\begin{prop} \label{prop0}  A \em uniform \em strategy is defined for $|U|>0$ to satisfy $U(x_j) = \frac{|U|}{M+1}$ for all $j$.  Then for any strategy $A$, with $|A|>0$ and $\mca(A) \leq \frac 1 2$, we have 

\[ \wp(A; U) \leq 0, \quad \wp(U; A) \geq 0,\] 
with equality if and only if $\mca(A) = \frac 1 2$.
\end{prop} 

\begin{proof} 
The payoff 

\[ \wp(A; U) = \frac{1}{|A|+|U|} \sum_{j=0} ^M A(x_j) \left( \sum_{i<j} U(x_i) - \sum_{i>j} U(x_i) \right).\]
We note that 

\[ \sum_{i>j} U(x_i) = |U| - U(x_j) - \sum_{i<j} U(x_i) = |U| - \frac{|U|}{M+1} - j \frac{|U|}{M+1},\] 
so the payoff is 

\[   \wp(A; U) = \frac{1}{|A|+|U|} \sum_{j=0} ^M A(x_j) \left(2 j \frac{|U|}{M+1} - |U| + \frac{|U|}{M+1}\right)\]

\[ =  \frac{1}{|A|+|U|} \left( \frac{2M|U|}{M+1} \sum_{j=0} ^M A(x_j) \frac{j}{M} - |U||A| + \frac{|U||A|}{M+1} \right)\] 

\[ =  \frac{1}{|A|+|U|} \left( \frac{2 M |A| |U|}{M+1} \mca(A) - |U| |A| + \frac{|U||A|}{M+1} \right) \] 

\[ = \frac{|A||U|}{(M+1)(|A|+|U|)} \left( 2 M \mca(A) - (M+1) + 1 \right) \] 

\[ = \frac{M |A||U|}{(M+1)(|A|+|U|)} \left( 2 \mca(A) - 1 \right) \leq 0 \textrm{ since } \mca(A) \leq \frac 1 2,\] 
and we see that equality holds precisely when $\mca(A) = \frac 1 2$. 
\end{proof} 

\begin{prop} \label{prop1} Assume that a strategy $A$ in the discrete case has $\mca(A) = 0.5$ and is not symmetric about $0.5$, then there exists at least one $\ell$ with $0 \leq \ell \leq M$ such that 

\[ A(x_\ell) + 2 \sum_{j=0} ^\ell A(x_j) > A(x_{M-\ell}) + 2 \sum_{j=M-\ell + 1} ^M A(x_j).\] 
\end{prop} 
\begin{proof} The proof is by contrapositive.  We shall assume that no such $\ell$  exists.  Then, we will show that the strategy, $A$, is symmetric about the competitive ability $0.5$.  So, we assume that for all $\ell$ we have the inequality:

\begin{equation} \label{eq:propsymmetric} A(x_\ell) + 2 \sum_{j=0} ^\ell A(x_j) \leq A(x_{M-\ell}) + 2 \sum_{j=M-\ell + 1} ^M A(x_j). \end{equation} 
 Whenever the sum is empty, we define its value to be equal to zero.  We claim that this inequality must be an equality for all $\ell$.  To see this, we shall sum over all $\ell$.  On the left we obtain 

 \[ \sum_{\ell = 0} ^M 2 \sum_{j<\ell} A(x_j) + \sum_{\ell=0} ^M A(x_\ell)  = \sum_{k=0} ^M [2(M-k) + 1] A(x_k).\] 
 On the right, 

 \[ \sum_{\ell=0} ^M 2 \sum_{j >M-\ell} A(x_j) + \sum_{\ell=0} ^M A(x_{M-\ell}) = \sum_{k=0} ^M (2k+1) A(x_k).\] 
Hence, we have the inequality 

\[ \sum_{k=0} ^M [2(M-k) + 1] A(x_k) \leq \sum_{k=0} ^M (2k+1) A(x_k).\] 
 This reduces to 

 \[ \sum_{k=0} ^M 2M A(x_k) = 2 M |A| \leq \sum_{k=0} ^M 4k A(x_k).\]
 Since the MCA of A is assumed to be equal to $\frac 1 2$,

 \[ \sum_{k=0} ^M 4k A(x_k) = 4 M \frac{|A|}{2} = 2 M |A|.\]
 Hence, both sides are equal.  Therefore, the inequality \eqref{eq:propsymmetric} can never be strict, because once we sum over all $0 \leq \ell \leq M$, the result is an equality.  Hence, we have for all $0 \leq \ell \leq M$, 

 \begin{equation} \label{eq:propsymmeq} 2 \sum_{j < \ell} A(x_j) + A(x_\ell) = 2 \sum_{j > M-\ell} A(x_j) + A(x_{M - \ell}). \end{equation} 
The proof shall be completed by induction on $\ell$.  For the base case, $\ell=0$, so the left side of \eqref{eq:propsymmeq} is $A(x_0)$, and the right side is $A(x_M)$, and these must be equal. Next, we inductively assume that $A(x_j) = A(x_{M-j})$ holds true for all $j=0, \ldots, \ell$ for some $\ell \geq 0$ with $\ell \leq M-1$.  Then, for $\ell+1 \leq M$ we have by \eqref{eq:propsymmeq} 

\[ 2 \sum_{j < \ell+1} A(x_j) + A(x_{\ell+1}) = 2 \sum_{j > M-\ell - 1} A(x_j) + A(x_{M-\ell-1}).\]
Since 

\[ 2 \sum_{j < \ell} A(x_j) + A(x_\ell) = 2 \sum_{j > M-\ell} A(x_j) + A(x_{M-\ell}),\] 

\[ 2 A(x_\ell) + A(x_{\ell+1}) = 2 A(x_{M-\ell}) + A(x_{M-\ell-1}).\] 
By the induction assumption, $A(x_\ell) = A(x_{M-\ell})$, so we obtain $A(x_{\ell+1}) = A(x_{M-\ell-1})$.
This completes the proof by induction. 
\end{proof} 

We use the preceding proposition to prove that strategies that are not symmetric about $\frac 1 2$ can be defeated.  First we note that for all strategies, there is a zero-sum dynamic, namely, 

\[ \wp(A; B) + \wp(B; A) = 0.\] 

\begin{prop} \label{prop2} Let $A$ be a strategy that has $\mca(A) = \frac 1 2$ that is not symmetric with respect to $\frac 1 2$.  Then there exists a strategy $B$ that has $\mca(B) = \frac 1 2$ for which 

\[ \wp(A; B) < 0, \quad \wp(B; A) > 0.\] 
\end{prop} 

\begin{proof}
By the preceding proposition, since A is not symmetric with respect to $0.5$, we have proven that there is an $\ell$ such that

\begin{equation} \label{eq:prop2_1} 2 \sum_{j < \ell} A(x_j) + A(x_\ell) > 2 \sum_{j > M-\ell} A(x_j) + A(x_{M-\ell}).  \end{equation} 
Whenever the sum is empty, define its value to be zero.  We define 

\[ B(x_j) := \begin{cases} \frac{|A|}{2} & j \in \{\ell, M-\ell \} \\ 0 & \textrm{ otherwise.}\} \end{cases} \]   
The mean competitive ability of $B$, 

\[  \mca(B) = \frac{1}{|A|} \left( \frac{\ell}{M} \frac{|A|}{2} + \frac{M-\ell}{M} \frac{|A|}{2} \right) = \frac 1 2.\] 
We compute the payoff

\[ \wp(B; A) = \frac{1}{2|A|} \left( B(x_\ell) \left( \sum_{j < \ell} A(x_j) - \sum_{j>\ell} A(x_j) \right) + B(x_{M-\ell}) \left(  \sum_{j < M- \ell} A(x_j) - \sum_{j>M-\ell} A(x_j) \right) \right)\]

\[ = \frac{1}{4}  \left( \left( \sum_{j < \ell} A(x_j) - \sum_{j>\ell} A(x_j) \right) + \left(  \sum_{j < M- \ell} A(x_j) - \sum_{j>M-\ell} A(x_j) \right) \right).\]
We note that 

\[ \sum_{j>\ell} A(x_j) = |A| - A(x_\ell) - \sum_{j < \ell} A(x_j), \quad \sum_{j < M-\ell} A(x_j) = |A| - A(x_{M-\ell}) - \sum_{j > M-\ell} A(x_j).\] 
This allows us to re-write $\wp(B; A)$  

\[ = \frac 1 4  \left( \sum_{j < \ell} A(x_j) - \left(  |A| - A(x_\ell) - \sum_{j < \ell} A(x_j) \right) + \left(   |A| - A(x_{M-\ell}) - \sum_{j > M-\ell} A(x_j) - \sum_{j>M-\ell} A(x_j) \right) \right) \] 

\[ = \frac 1 4 \left( 2 \sum_{j < \ell} A(x_j) + A(x_\ell) - \left( 2 \sum_{j > M-\ell} A(x_j) + A(x_{M-\ell}) \right)  \right) > 0,\]
 with the final inequality following from \eqref{eq:prop2_1}.
\end{proof} 

\begin{prop} \label{prop3}  Assume that $(A, B)$ is an equilibrium strategy.  Then 

\[ \wp(A; B) = \wp(B; A) = 0. \] 
Moreover, we have for any strategy $C$, 

\begin{equation} \label{eq:eq_strategy_2} \wp(*; C) \geq 0, \quad * = A, B. \end{equation} 
In case $M$ is odd, all equilibrium strategies are uniform.  In case $M$ is even, all equilibrium strategies have $\mca$ equal to $\frac 1 2$ and further satisfy $A(x_{2j}) = A(x_0)$, $A(x_{2j+1}) = A(x_1)$ for all $j=0, 1, \ldots, \frac M 2$.  Furthermore equality holds in  \eqref{eq:eq_strategy_2} if and only if $\mca(C) = \frac 1 2$. 
\end{prop} 

\begin{proof} 
To prove the first statement, we note that if 

\[ \wp(A; B) < 0 \implies \wp(B; B) = 0 > \wp(A; B),\] 
contradicting the definition of equilibrium strategy.  Thus $\wp(A; B) \geq 0$.  The same argument shows that $\wp(A; B) \geq 0$, so by the zero-sum dynamic, $\wp(A; B) = \wp(B; A) = 0$.  To prove the second statement, assume that there is a strategy $C$ such that 

\[ \wp(A; C) < 0.\] 
Then, by the zero-sum dynamic

\[ \wp(C; A) > 0 = \wp(B; A),\] 
contradicting the definition of $B$ as an equilibrium strategy.  The same argument proves that \eqref{eq:eq_strategy_2} holds for $B$ as well.  By Proposition \ref{prop0}, any strategy with $\mca < \frac 1 2$ is not an equilibrium strategy.  By Proposition \ref{prop2}, any strategy that is not symmetric with respect to $\frac 1 2$ is not an equilibrium strategy.  By Proposition \ref{prop0}, strategies comprised of uniform strategies $U_1$ and $U_2$ satisfy 

\[ \wp(U_1; U_2) = 0 = \wp(U_2; U_1), \quad \wp(C; U_i) = - \wp(U_i; C) \leq 0,\] 
for any strategy $C$, with equality if and only if $\mca(C) = \frac 1 2$.  Consequently, any such $(U_1, U_2)$ is an equilibrium strategy. 
Let us now see that in case $M$ is odd, all equilibrium strategies are comprised of uniform strategies.  For this aim, it suffices to consider strategies that have $\mca=\frac 1 2$ and are symmetric about $\frac 1 2$.  We will show that such a strategy $A$ that is not uniform cannot be an equilibrium strategy.  Let $\ell$ be the smallest integer such that $A(x_{\ell+1}) = \ldots = A(x_{M-\ell-1})$.  Since we have assumed that $A$ is symmetric,

\[ A \left( x_{\frac{M-1}{2}} \right) = A\left( x_{\frac{M+1}{2}} \right).\]
Consequently $\ell < \frac{M-1}{2}$, and by the assumption that $A$ is not uniform, there exists such an $\ell \geq 0$.  There are only two cases to consider:

\[ \textrm{Case 1:  } A(x_\ell) > A(x_{\ell+1}), \quad \textrm{Case 2:  } A(x_{\ell}) < A(x_{\ell+1}). \] 
Assume first that we are in the first case.  Let 

\[ B(x_j) :=  \begin{cases} |A| \frac{M-2\ell}{M-2\ell+1} & j=\frac{M-1}{2}, \\ |A| \frac{1}{M-2\ell+1} & j=M-\ell, \\ 0 & \textrm{ otherwise.} \end{cases}  \] 
Then we compute 

\[ \mca(B) = \frac{M-1}{2M} \frac{M-2\ell}{M-2\ell+1} + \frac{M-\ell}{M} \frac{1}{M-2\ell+1} = \frac{(M-1)(M-2\ell) + 2(M-\ell)}{2M(M-2\ell+1)} = \frac 1 2.\] 
We compute the payoff $\wp(B; A) =$ 

\[ \frac{1}{2|A|}  |A| \frac{M-2\ell}{M-2\ell+1} \left( \sum_{i < \frac{M-1}{2}} A(x_i) - \sum_{i>\frac{M-1}{2}} A(x_i) \right) \] 

\[ + \frac{1}{2|A|} |A| \frac{1}{M-2\ell+1}  \left( \sum_{i < M-\ell} A(x_i) - \sum_{i>M-\ell} A(x_i) \right). \] 

\[ = \frac{1}{2(M-2\ell+1)} \left((M-2\ell) \left( \sum_{i < \frac{M-1}{2}} A(x_i) - \sum_{i>\frac{M-1}{2}} A(x_i) \right) +  \left( \sum_{i < M-\ell} A(x_i) - \sum_{i>M-\ell} A(x_i) \right) \right). \] 
By the symmetry assumption, this is 

\[ =  \frac{1}{2(M-2\ell+1)} \left( - (M-2\ell)A(x_{\frac{M+1}{2}}) +  \left( \sum_{i < M-\ell} A(x_i) - \sum_{i>M-\ell} A(x_i) \right) \right). \] 
By the definition of $\ell$, this is 

\[ =  \frac{1}{2(M-2\ell+1)} \left( - (M-2\ell)A(x_{\ell+1}) +  \left( \sum_{i < M-\ell} A(x_i) - \sum_{i>M-\ell} A(x_i) \right) \right). \] 
Using the symmetry of $A$ and the definition of $\ell$ this is 

\[ =  \frac{1}{2(M-2\ell+1)} \left( - (M-2\ell)A(x_{\ell+1}) +  (M-2\ell-1) A(x_{\ell+1}) + A(x_\ell) \right)\] 

\[=  \frac{1}{2(M-2\ell+1)} (A(x_\ell) - A(x_{\ell+1})) >0,\]
with the last inequality following from the fact that we are in Case 1.   

Next we assume that we are in Case 2, so we have $A(x_\ell) < A(x_{\ell+1})$.  We define 

\[ B(x_j) := \begin{cases} |A| \frac{1}{M-2\ell+1} & j=\ell, \\ |A| \frac{M-2\ell}{M-2\ell+1} & j= \frac{M+1}{2}, \\ 0 & \textrm{ otherwise.} \end{cases} \] 
Then we compute 

\[ \mca(B) = \frac \ell M \frac{1}{M-2\ell + 1} + \frac{M+1}{2M} \frac{M-2\ell}{M-2\ell+1} = \frac{2 \ell + (M+1)(M-2\ell)}{2M(M-2\ell+1)} = \frac 1 2. \] 
By definition, $\wp(B; A) =$

\[ \frac{1}{2|A|}  |A| \frac{1}{M-2\ell+1} \left( \sum_{j<\ell} A(x_j) - \sum_{j>\ell} A(x_j) \right) \] 

\[ +  \frac{1}{2|A|}  |A| \frac{M-2\ell}{M-2\ell + 1} \left( \sum_{j<\frac{M+1}{2}} A(x_j) - \sum_{j > \frac{M+1}{2}} A(x_j) \right).\] 

\[ = \frac{1}{2(M-2\ell + 1)}\left( - (M-2\ell-1) A(x_{\ell+1}) - A(x_\ell) + (M-2\ell) A(x_{\frac{M-1}{2}}) \right)\] 

\[ = \frac{1}{2(M-2\ell + 1)}\left(  -(M-2\ell-1) A(x_{\ell+1}) - A(x_\ell) + (M-2\ell) A(x_{\ell+1}) \right)\] 

\[ = \frac{1}{2(M-2\ell+1)} (A(x_{\ell+1}) - A(x_\ell)) >0.\] 
Above we have used the symmetry of $A$ and the definition of $\ell$, and the fact that we are in Case 2.  This completes the proof that in case $M$ is odd, all equilibrium strategies are comprised of uniform strategies.  

Let us now assume that $M$ is even.  In this case we already have demonstrated that equilibrium strategies must have $\mca = \frac 1 2$ and must be symmetric.  The last part of the proof is to demonstrate that in this case, equilibrium strategies must further satisfy the condition $A(x_{2j}) = A(x_0)$ and $A(x_{2j+1}) = A(x_1)$ for all $1 \leq j \leq \frac M 2$.  This will proceed in two parts.  First we show that a strategy that has $\mca = \frac 1 2$, is symmetric, but does not satisfy this condition cannot be an equilibrium strategy.  We therefore assume that there is some $2 \leq \ell \leq \frac M 2$ such that $A(x_{\ell -2}) \neq A(x_{\ell})$.  There are two cases: 

\[ \textrm{Case 1:  }  A(x_{\ell-2}) < A(x_\ell), \quad \textrm{ Case 2:  } A(x_{\ell-2}) > A(x_\ell). \] 
Assume we are in the first case.  We define 

\[ B(x_j) := \begin{cases} \frac{|A|}{2} & j=M-\ell + 1, \\ \frac{|A|}{4} & j=\ell, \\ \frac{|A|}{4} & j = \ell - 2. \end{cases} \] 
Then 

\[ \mca(B) = \frac{1}{|A|} \left( \frac{|A|}{2} \frac{M-\ell+1}{M} + \frac{|A|}{4} \frac{\ell}{M} + \frac{|A|}{4} \frac{\ell-2}{M}\right)\] 

\[= \frac{2M-2\ell+2+\ell+\ell-2}{4M} = \frac 1 2.\]   
The payoff $\wp(B; A) = $

\[ \frac{1}{2|A|} \frac{|A|}{4} \left( \sum_{j<\ell-2} A(x_j) - \sum_{j>\ell-2} A(x_j) + \sum_{j<\ell} A(x_j) - \sum_{j>\ell} A(x_j) \right)\] 

\[+\frac{1}{2|A|} \frac{|A|}{2} \left( \sum_{j<M-\ell+1} A(x_j) -  \sum_{j>M-\ell+1} A(x_j) \right)\] 

\[ = \frac{1}{8} \left( \sum_{j < \ell - 2} A(x_j) - A(x_{\ell-1}) - A(x_\ell) - \sum_{\ell<j\leq M-\ell} A(x_j) - A(x_{M-\ell+1}) - A(x_{M-\ell+2})) \right)\] 

\[ + \frac 1 8 \left(  - \sum_{j>M-\ell+2} A(x_j) + \sum_{j < \ell} A(x_j) - \sum_{\ell < j \leq M-\ell} A(x_j) - \sum_{j>M-\ell} A(x_j) \right) \] 

\[ + \frac 1 4 \left( \sum_{j<\ell-1} A(x_j) + A(x_{\ell-1}) + A(x_\ell) + \sum_{\ell < j \leq M-\ell} A(x_j) - \sum_{j > M-\ell+1} A(x_j) \right).\] 
Note that by symmetry, $A(x_j) = A(x_{M-j})$.  Therefore this is 

\[ \frac 1 8 \left( - 2 A(x_{\ell-1}) - A(x_\ell) - A(x_{\ell-2}) - \sum_{\ell < j \leq M-\ell} A(x_j) - \sum_{\ell < j \leq M-\ell} A(x_j)\right) \] 

\[ + \frac 1 4 \left( A(x_{\ell-1}) + A(x_\ell) + \sum_{\ell < j \leq M-\ell} A(x_j) \right) = \frac 1 8\left( A(x_\ell) - A(x_{\ell-2})\right) >0. \] 
The last inequality follows because we are in Case 1.  Now let us assume that we are in Case 2, so $A(x_{\ell-2}) > A(x_\ell)$.  In this case we define 

\[ B(x_j) = \begin{cases} \frac{|A|}{2} & j=\ell-1, \\ \frac{|A|}{4} & j=M-\ell, \\ \frac{|A|}{4} & j=M-\ell+2. \end{cases} \] 
Then, 

\[ \mca(B) = \frac{1}{|A|} \left( \frac{|A|}{2} \frac{\ell-1}{M} + \frac{|A|}{4} \frac{M-\ell}{M} + \frac{|A|}{4} \frac{M-\ell+2}{<}\right)\] 

\[ = \frac{2\ell-2+M-\ell + M-\ell+2}{4 M} = \frac 1 2.\] 
We compute that the payoff $\wp(B; A)=$

\[ \frac{1}{2|A|} \frac{|A|}{2} \left( \sum_{j < \ell-1} A(x_j) - \sum_{\ell \leq j < M-\ell} A(x_j) - A(x_{M-\ell}) - A(x_{M-\ell+1}) - \sum_{j>M-\ell + 1} A(x_j) \right) \] 

\[ + \frac{1}{2|A|} \frac{|A|}{4} \left( \sum_{j < \ell} A(x_j) + \sum_{\ell \leq j < M-\ell} A(x_j) - \sum_{j > M-\ell} A(x_j) \right) \] 

\[ + \frac{1}{2|A|} \frac{|A|}{4} \left( \sum_{j < \ell-2} A(x_j) + A(x_{\ell-2}) + A(x_{\ell-1}) + \sum_{\ell \leq j < M-\ell} A(x_j) + A(x_{M-\ell}) + A(x_{M-\ell+1})\right)\]

\[ - \frac{1}{2|A|} \frac{|A|}{4} \sum_{j>M-\ell+2} A(x_j).\]   

By the symmetry, $A(x_j) = A(x_{M-j})$ and so this is 

\[ \frac 1 4 \left( - \sum_{\ell \leq j < M-\ell} A(x_j) - A(x_{\ell}) - A(x_{\ell-1}) \right)\] 

\[ + \frac 1 8 \left( \sum_{\ell \leq j < M-\ell} A(x_j) + A(x_{\ell-2}) + A(x_\ell) + 2 A(x_{\ell-1}) + \sum_{\ell \leq j < M-\ell} A(x_j) \right)\] 

\[ = \frac 1 8 \left( A(x_{\ell-2}) - A(x_\ell) \right) > 0.\] 
The last inequality follows because we are in Case 2.  

To complete the proof, we will demonstrate that $\wp(A;B) \geq 0$ for all strategies $B$ for any strategy $A$ that satisfies: 
\begin{enumerate} 
\item $\mca(A) = \frac 1 2$;
\item $A(x_j) = A(x_{M-j})$ for all $j$;
\item for all $2 \leq j \leq \frac M 2$, $A(x_j) = A(x_{j-2})$.
\end{enumerate} 
Moreover, we will prove that $\wp(A;B) = 0$ if and only if $\mca(B) = \frac 1 2$.  Consequently, in case $M$ is even, all equilibrium strategies are comprised of strategies of this type.  Assuming that $A$ satisfies the conditions above, there exist non-negative numbers, $a$ and $b$ such that 

\[ A(x_{2j}) = a, \quad A(x_{2j+1}) = b, \quad \textrm{ for all } j=0, \ldots, \frac M 2.\] 
We shall compute the payoff for an arbitrary strategy $B$ in competition against such a strategy $A$, 

\[ \wp(B; A) = \frac{1}{|A| + |B|} \sum_{i=0} ^M B(x_i) \left( \sum_{j=0} ^{i-1} A(x_j) - \sum_{j=i+1} ^M A(x_j) \right).\] 
To exploit the symmetry of $A$, we write this as 

\[ \frac{1}{|A| + |B|} \sum_{i=0} ^{M/2} B(x_i) \left( \sum_{j=0} ^{i-1} A(x_j) - \sum_{j=i+1} ^M A(x_j) \right) \] 

\[ + \frac{1}{|A| + |B|} \sum_{i=\frac M 2 + 1} ^ M B(x_i) \left(\sum_{j=0} ^{i-1} A(x_j) - \sum_{j=i+1} ^M A(x_j) \right)\] 

\[ = - \frac{1}{|A| + |B|} \sum_{i=0} ^{M/2} B(x_i) \sum_{j=i+1} ^{M-i} A(x_j) + \frac{1}{|A| + |B|} \sum_{i=\frac M 2 + 1} ^ M B(x_i) \sum_{j=M-i} ^{i-1} A(x_j),\]
having used the symmetry of $A$ above.  The first sum over $j$, $\sum_{j=i+1} ^{M-i} A(x_j)$, has $M-2i$ summands, half of which are equal to $a$, and half of which are equal to $b$.  Similarly, the second sum over $j$, $\sum_{j=M-i} ^{i-1} A(x_j)$, has $2i-M$ summands, half of which are equal to $a$, and half of which are equal to $b$.  Therefore we obtain:  

\[ \wp(B; A) = \frac{1}{|A| + |B|} \left[ - \sum_{i=0} ^{M/2} B(x_i) \left( \frac M 2 - i \right) (a+b) + \sum_{i=\frac M 2 + 1} ^M B(x_i) \left( i - \frac M 2 \right) (a+b)\right] \] 

\[ = \frac{1}{|A|+|B|} (a+b) \sum_{i=0} ^M B(x_i) \left( i - \frac M 2 \right) = \frac{M (a+b)}{|A|+|B|} \left( - \frac{|B|}{2} + \sum_{i=0} ^M \frac i M B(x_i) \right)\]

\[ =  \frac{M (a+b)}{|A|+|B|} \left( - \frac{|B|}{2} + |B| \mca(B) \right) \leq 0,\] 
with equality if and only if $\mca(B) = \frac 1 2$.  Consequently we have proven that for such an $A$, we have 

\[ \wp(B;A) \leq 0 \textrm{ for any $B$ with equality if and only if $\mca(B) = \frac 1 2$}\] 
that immediately implies 

\[ \wp(A; B) \geq 0 \textrm{ for any $B$ with equality if and only if $\mca(B) = \frac 1 2$}.\] 
We therefore obtain that in case $M$ is even, equilibrium strategies are precisely those that are comprised of $A$ that satisfy these three conditions.  
\end{proof} 

We will now use the results obtained here to determine all equilibrium strategies in the discrete model for arbitrary numbers of competing species.  

\subsubsection*{Proof of Theorem 1 for the discrete model} \label{si:discrete_t1} 
\begin{proof}
Assume that $\{A_k\}_{k=1} ^n$ is a set of strategies that satisfies the assumptions of the theorem that characterize the equilibrium strategies.   We begin by proving that these characteristics do ensure that this set of strategies is an equilibrium strategy.  Consider any strategy $B$ as in Definition 1.  Then if we change strategy $A_1$ to $B$, the resulting payoff $\wp (B; A_2, \ldots, A_n) = $

\[  \frac{1}{|B| + \sum_{k \geq 2} |A_k|} \sum_{j=0} ^M B(x_j) \left[ \sum_{i < j} \left( B(x_i) + \sum_{k \geq 2} A_k (x_i) \right) - \sum_{i > j} \left(B(x_i) + \sum_{k \geq 2} A_k (x_i) \right) \right].\] 
However, by the assumption on $A_k$ in the theorem and the zero-sum dynamic, we have 

\[\wp(A_k; B) \geq 0 \implies \wp(B; A_k) \leq 0 \implies \sum_{j=0} ^M B(x_j) \left[ \sum_{i<j} A_k (x_i) - \sum_{i>j} A_k (x_i) \right] \leq 0,\] 
for all $k=2, \ldots, n$.  
Therefore, since the internal competition within $B$ does not affect the payoff, we have 

\begin{equation} \label{eq:proof1} \sum_{j=0} ^M B(x_j) \left[ \sum_{i<j} \sum_{k=2} ^n A_k (x_i) - \sum_{i>j} \sum_{k=2} ^n A_k (x_i) \right] \leq 0 \implies \wp (B; A_2, \ldots, A_n) \leq 0. \end{equation} 

By the assumptions of the theorem we have for all pairs $A_j$ and $A_k$ both 

\[ \wp(A_k; A_j) \geq 0, \quad \wp(A_j; A_k) \geq 0. \] 
Due to the zero sum dynamic

\[ \wp(A_k; A_j) = - \wp(A_j; A_k) \implies \wp(A_k; A_j) = 0 \quad \forall j, k \in \{1, \ldots, n\}.\] 
Consequently, we have

\[ \sum_{j=0} ^M A_k (x_j) \left[ \sum_{i<j} A_\ell(x_i) - \sum_{i>j} A_\ell(x_i) \right] = 0, \quad \forall k, \ell \in \{1, \ldots, n\},\] 
and therefore 

\[ \wp(A_k; A_1, \ldots, A_{k-1}, A_{k+1}, A_n) = 0 \quad \forall k=1, \ldots, n.\] 
By \eqref{eq:proof1}, we therefore have $\wp(B; A_2, \ldots, A_n) \leq \wp(A_1; A_2, \ldots, A_n) = 0$.   The same argument applies to $A_k$ for all $k$, namely, for any $B$ as in, Definition 1, 

\[ \wp(B; A_1, \ldots, A_{k-1}, A_{k+1}, \ldots A_n) \leq 0 = \wp(A_k; A_1, \ldots, A_{k-1}, A_{k+1}, \ldots, A_n).\] 
It follows that the set of strategies $\{A_k\}_{k=1} ^n$ is an equilibrium strategy.  

To prove the converse, we begin by assuming that $\{A_k\}_{k=1} ^n$ is an equilibrium strategy.  For the sake of contradiction, let us see what would happen if 

\[ \wp(A_1; A_2 \ldots, A_n) < 0 \implies \sum_{j=0} ^M A_1 (x_j) \left[ \sum_{i<j} \sum_{k=2} ^n A_k (x_i) - \sum_{i<j} \sum_{k=2} ^n A_k (x_i) \right] < 0.\] 
Defining a new strategy $B$ so that 

\[ B(x_j) := \sum_{k=2} ^n A_k (x_j) \implies \mca(B) \leq \frac 1 2,\]
and 

\[ \sum_{j=0} ^M B(x_j) \left[ \sum_{i<j} \sum_{k=2} ^n A_k (x_i) - \sum_{i>j} \sum_{k=2} ^n A_k (x_i) \right] = 0 \implies \wp(B; A_2, \ldots, A_n) = 0 > \wp (A_1; \ldots, A_n),\] 
contradicting the definition of equilibrium strategy.  We therefore have that $\wp(A_1; A_2, \ldots, A_n) = 0$, and the same argument shows that $\wp (A_k; A_1, \ldots, A_{k-1}, A_{k+1}, \ldots A_n ) = 0$ for all $k=1, \ldots, n$.  By the definition of equilibrium strategy, we must have that for any strategy $B$, 

\[ \wp (B; A_2, \ldots, A_n) \leq \wp (A_1; A_2 \ldots, A_n) = 0 \implies \sum_{j=0} ^M B(x_j) \left[ \sum_{i<j} \sum_{k=2} ^n A_k (x_i) - \sum_{i>j} \sum_{k=2} ^n A_k (x_i) \right] \leq 0.\] 
We define a new strategy 

 \[ B_2 (x_j) := \sum_{k=2} ^n A_k (x_j) \implies \mca(B_2) \leq \frac 1 2.\] 
Then the payoff to any strategy $B$ in competition with $B_2$ satisfies 

\[ \wp (B; B_2) \leq 0.\] 
By the zero sum dynamic, we therefore have 

\[ \wp (B_2; B) \geq 0.\] 
By the definition of equilibrium strategy, since $\wp (A_k; A_1, \ldots , A_{k-1}, A_{k+1}, \ldots, A_n) = 0$ for all $k$, we define strategies 

\[ B_k (x_j) := \sum_{\ell \neq k} A_\ell (x_j),\] 
that all satisfy 

\[ \wp (B_k; B) \geq 0\] 
for all strategies $B$ as in Definition 1.  We therefore have that 

\[ \wp(B_k; B) \geq 0, \] 
with equality if and only if $\mca(B) = \frac 1 2$.  Consequently, we obtain that for all $B$ with $\mca(B) = \frac 1 2$, 

\[ \wp(B; B_k) = 0 \implies \sum_{k=1} ^n \wp(B; B_k) = 0 \implies (n-1) \sum_{j=0} ^M B(x_j) \left[ \sum_{i<j} \sum_{k=1} ^n A_k (x_i) - \sum_{i>j} \sum_{k=1} ^n A_k (x_i) \right] = 0,\] 
as well as 

\[ \wp(B; B_1) = 0 \implies \sum_{j=0} ^M B(x_j) \left[ \sum_{i<j} \sum_{k=2} ^n A_k (x_i) - \sum_{i>j} \sum_{k=2} ^n A_k (x_i) \right] = 0.\] 
We therefore obtain that for all $B$ with $\mca(B) = \frac 1 2$, 

\[ \sum_{j=0} ^M B(x_j) \left[ \sum_{i<j} A_1 (x_i) - \sum_{i > j} A_1 (x_i) \right] = 0 \implies \wp(B; A_1) = 0.\] 
It follows that in fact for all strategies $B$ as in Definition 1 we have 

\[ \wp(A_1; B) \geq 0, \] 
with equality if and only if $\mca(B) = \frac 1 2$.  The same argument shows that each $A_k$ satisfies 

\[ \wp(A_k; B) \geq 0\] 
for all $B$ as in Definition 1. 
\end{proof} 

\subsection*{Mathematical proofs for the continuous model} \label{si:continuous} 
Here we determine all equilibrium strategies for the continuous model when two strategies compete.  We will complete the proof of the main theorem for \em both the continuous and discrete models simultaneously \em in the following section.  
\begin{prop} \label{prop1:cont}  In the continuous model, a pair of strategies $(f_1, f_2)$ is an equilibrium strategy if and only if both $f_1$ and $f_2$ are positive constants.  Moreover, they satisfy 
\[ \wp(f_i; g) \geq 0, \quad i = 1,2\] 
for all strategies $g$ with equality if and only if $\mca(g) = \frac 1 2$. 
\end{prop} 

\begin{proof} 
For any $g$ competing with a positive constant function, denoted by $u(x)=U(1)$,
	
	\[ \wp(u; g) = \frac{1}{U(1) + G(1)} \int_0^1 U(1) \left( \int_0 ^x g(t) dt - \int_x ^1 g(t) dt \right) dx \] 
	
	\[ = \frac{U(1)}{U(1) + G(1)} \int_0 ^1 \left( G(x) - G(1) + G(x) \right) dx = \frac{U(1)}{U(1) + G(1)} \left( \int_0 ^1 2 G(x)dx  - G(1) \right) , \]
	 where $G(x) := \int_0 ^x g(t) dt$.  We recall that the constraint requires 
	
	 \[ \frac{1}{G(1)} \int_0 ^x x g(x) dx \leq \frac 1 2.\] 
	 Integration by parts shows that this is equivalent to 
	
	 \[ \frac{1}{G(1)} \left( G(1) - \int_0 ^1 G(x) dx \right) \leq \frac 1 2 \iff G(1) \leq 2 \int_0 ^1 G(x) dx. \] 
	 Consequently, we obtain that 
	 
	 \[\wp(u;g) \geq 0,\]
	 with equality precisely when the $\mca$ constraint is an equality, that is $\mca(g) = \frac 1 2$.  
	
	We therefore obtain that for $f_1$ and $f_2$ as in the theorem, 
	
	\[ \wp(f_1; f_2) = 0 = \wp(f_2; f_1).\] 
	Moreover, changing the strategy 
	
	\[ \wp(g; f_2) = - \wp(f_2; g) \leq 0 = \wp(f_1; f_2)\] 
	for any strategy $g$, and similarly, 
	
	\[ \wp(g; f_1) = - \wp(f_1; g) \leq 0 = \wp(f_2; f_1).\]
	The pair $(f_1, f_2)$ is therefore by definition an equilibrium strategy.  To characterize all equilibrium strategies, we note that if $(f,g)$ is an equilibrium strategy, then we must have 
	
	\[ \wp(f;g) = \wp(g;f) = 0.\] 
	To see why this must be true, assume for the sake of contradiction that $\wp(f;g) > 0$.  Then by the zero sum dynamic 
	
	\[ \wp(g; f) < 0 = \wp(f; f).\] 
	This contradicts the definition of equilibrium strategy.  Next, assume for the sake of contradiction that $\mca(g) < \frac 1 2$.  Then we have 
	
	\[ \wp(f;g) = 0 < \wp(u;g),\] 
	contradicting the definition of equilibrium strategy.  It follows that all strategies that comprise an equilibrium strategy must have $\mca = \frac 1 2$.  Moreover, since an equilibrium strategy $(f_1, f_2)$ must satisfy 
	
	\[ \wp(f_1; f_2) = \wp(f_2; f_1) = 0,\] 
	by definition of equilibrium strategy there can not exist a strategy $g$ such that either 
	
	\[ \wp(g; f_2) > 0 \textrm{ or } \wp(g; f_1) > 0.\] 
	To prove that all equilibrium strategies are comprised of positive constant functions, we will show that any strategy that for any strategy $f$ that is not a positive constant, we can construct a strategy $g$ such that 

	\begin{equation} \label{eq:g_beats_f}  \wp(g; f) > 0. \end{equation} 

	If $\mca(f)<1/2$ then $g(x)=G(1)>0$ satisfies \eqref{eq:g_beats_f}, so we may assume $\mca(f)=1/2$.  Then, we have that 

	\[ \wp(g;f) > 0 \iff \int_0 ^1 g(x) (2 F(x) - F(1)) dx > 0 \iff \int_0 ^1 g(x) \left( F(x) - \frac{F(1)}{2} \right) dx > 0\] 

	\[ \iff \int_0 ^1 g(x) \left( \frac{F(x)}{F(1)} - \frac 1 2 \right) dx > 0.\] 
	We will search for $g$ that has $\mca(g) = \frac 1 2$, consequently
	
	\[ \int_0 ^1 \frac{g(x)}{2} dx = \frac{G(1)}{2} = \int_0 ^1 x g(x) dx.\] 
	We therefore substitute obtaining 
	
	\[ \wp(g; f) > 0 \iff  \int_0 ^1 g(x) \left( \frac{F(x)}{F(1)} -x \right) dx > 0.\]

	Since  $\mca(f)=0.5,$
	\[
	\int_0^1xdx=\frac{1}{2}=\frac{1}{F(1)}\int_0^1f(x)x dx=1-\int_0^1 \frac{F(x)}{F(1)}dx=2\int_0^1xdx-\int_0^1 \frac{F(x)}{F(1)}dx,
	\]
	that is,
	\[
	0=\int_0^1 \left(\frac{F(x)}{F(1)} - x \right)dx.
	\]
	Since $f$ is not constant there exists $x\in[0,1]$ such that $F(x)\neq xF(1)$. Thus, the integrand above must assume both positive and negative values, and in particular, 
	\begin{equation}\label{eq:notconstant}
	\exists \quad  0 \leq a < b \leq 1  \textrm{ such that }  \frac{F(x)}{F(1)}-x>0\quad \forall x\in[a,b].
	\end{equation}
	If $(a,b)$ is not fully contained in either $(0,1/2)$ or $(1/2,1)$, then by continuity, it is possible to split $(a,b)$ into smaller intervals, one of which is fully contained in either $(0,1/2)$ or $(1/2,1)$.  We therefore assume without loss of generality that $(a,b)$ is contained in either $(0,1/2)$ or $(1/2,1)$. 
	First, assume $(a,b)\subset (1/2,1)$.
	Define
	\begin{equation}\label{eq:g2}
	g(x)=
	\begin{cases}
	4M^2x, & x\in[0,1/(2M)]\\
	-4M^2x + 4M & x\in[1/(2M), 1/M ]\\  
	2N\frac{x-a}{b-a}, & x\in[a,(a+b)/2] \\
	2N\frac{b-x}{b-a}, & x\in[(a+b)/2,b] \\
	0,& \mathrm{otherwise.}
	\end{cases}
	\end{equation}
	The constants $M,N$ will be determined to guarantee that $\mca(g)=0.5$, and $\wp(g; f) > 0$.
	We denote by $I_1$ and $I_2$ the integrals
	\begin{equation}
	I_1=\int_a^b\left(\frac{F(x)}{F(1)}-x\right)g(x)dx, \qquad I_2=\int_0^{1/M}\left(\frac{F(x)}{F(1)}-x\right)g(x)dx.
	\end{equation}
	Since 
	\[ \wp(g; f) > 0 \iff \int_0 ^1 g(x) \left( \frac{F(x)}{F(1)} - x \right) dx > 0,\]
	and 
	\[\int_0 ^1 g(x) \left( \frac{F(x)}{F(1)} - x \right) dx  =  I_1 + I_2 \geq I_1 - |I_2|.\]
	we wish to estimate $I_1$ from below and $|I_2|$ from above.  
	
	By continuity, there is a strictly positive constant $R$ such that 
	\[ \frac{F(x)}{F(1)} - x > R > 0 \quad \forall x\in[a,b]. \] 
	Therefore, we have the estimate
	\[
	I_1\geq R N\frac{b-a}{2}.
	\] 
	By definition of $F$, we have the estimate 
	
	\[
	|F(x)|=\left|\int_0^x f(t)dt\right| \leq \frac{||f||_{\infty}}{M},\quad \forall x\in[0,1/M].
	\]
	Above, $||f||_{\infty}$ is the standard supremum norm of $f$ which is finite because $f$ is continuous on the compact set $[0, 1]$.  
	
	By the triangle inequality 
	\begin{equation*}
	\left|\frac{F(x)}{F(1)}-x\right|\leq \frac{1}{M} \left(\frac{||f||_{\infty}}{F(1)}+1\right),\qquad \forall x\in[0,1/M].
	\end{equation*}
	Since
	\[
	0\leq g(x)\leq 2M, \qquad \forall x\in[0,1/M].
	\]
	Thus, 
	\[
	|I_2|\leq \int_0^{1/M} 2M \frac{1}{M} \left(\frac{||f||_{\infty}}{F(1)}+1\right)dx=\frac{2}{M} \left(\frac{||f||_{\infty}}{F(1)}+1\right).
	\]
	To conclude, 
	\begin{equation}\label{eq:estimateApp}
	I_1-|I_2|\geq R N\frac{b-a}{2}-\frac{2}{M} \left(\frac{||f||_{\infty}}{F(1)}+1\right).
	\end{equation}
	We will choose $M,N$ such that both $\wp(g; f) > 0$ and $\mca(g)=1/2$, so we begin by computing 
	\[G(1)=1+N\frac{b-a}{2}\implies \frac{1}{G(1)}=\frac{2}{2+N(b-a)}.\] 
	Integrating $xg(x)$ gives
	\begin{equation}\label{eq:mcaAppG}
	\mca(g)=\frac{2}{2+N(b-a)}\left(\frac{N}{b-a}\frac{(b-a)(b^2-a^2)}{4}+\frac{1}{2M}\right).
	\end{equation}
	Set $\mca(g)=1/2$, that is,
	\begin{equation*}
	\frac{1}{2}= \frac{2}{2+N(b-a)}\left(N\frac{b^2-a^2}{4}+\frac{1}{2 M}\right).
	\end{equation*}
	Therefore,
	\begin{equation}\label{eq:Napp}
	N=\frac{2}{(b-a)(a+b-1)}\left(1- \frac{1}{M} \right).
	\end{equation}
	Here it is important to observe that since $a < b$, and $[a,b] \subset [1/2, 1]$, $a+b-1 > 0$, so we are not dividing by zero, and $N>0$.  

	Inserting $N$ from~(\ref{eq:Napp}) into~(\ref{eq:estimateApp}),
	\begin{equation*}
	I_1-|I_2|\geq R \frac{1}{a+b-1}\left(1-\frac{1}{M} \right)-\frac{2}{M} \left(\frac{||f||_{\infty}}{F(1)}+1\right).
	\end{equation*}
	Consequently $I_1>|I_2|$ is achieved by choosing $M$ such that
	\begin{equation}\label{eq:Mapp}
	2 \left( \frac{ ||f||_\infty }{F(1)} + 1 \right) \frac{a+b-1}{R} + 1  < M.
	\end{equation}

	This completes the proof in the case $(a,b)\subset(1/2,1)$. 
	
	If $(a,b)\subset (0,1/2)$, define
\begin{equation}\label{eq:g}
g(x)=
\begin{cases}
2N\frac{x-a}{b-a}, & x\in[a,(a+b)/2] \\
2N\frac{b-x}{b-a}, & x\in[(a+b)/2,b] \\
2M^2(x-1+1/M), & x\in[1-1/M,1-1/(2M)]\\
2M^2(1-x), & x\in[1-1/(2M),1]\\
0,& \mathrm{otherwise,}
\end{cases}
\end{equation}
with constants $M,N$ to be determined.  First, we shall fix $N$ so that the function $g$ satisfies the MCA constraint.  We compute 
\[ 
G(1) = 1 + N \frac{b-a}{2}, \quad 
\mca(g) = \frac{2}{2+N(b-a)}\left( N\frac{b^2-a^2}{4} +  1 - \frac{1}{2M} \right).
\] 
We set this equal to $\frac 1 2$, and obtain the equation
\[ 
N\frac{b^2-a^2}{4}+  1 - \frac{1}{2M} = \frac{1}{4}\big(2+N(b-a)\big).
\] 
Re-arranging, 
\[ 
N\frac{b-a}{4}(b+a) - N\frac{b-a}{4}=\frac{1}{2M}-\frac{1}{2},
\]
hence we require that 
\begin{equation} 
\label{eq:nmsymmetry} 
N=\frac{2}{b-a}\frac{1-1/M}{1-a-b}.
\end{equation} 
In this case, it is important to note that $(a,b) \subset (0, 1/2)$, so the denominator is positive.  We shall choose 
\begin{equation} 
\label{eq:mcondition1app} 
M>1. 
\end{equation} 
%
We denote by $I_1$ and $I_2$ the integrals
\begin{equation*}
I_1=\int_a^b\left(\frac{F(x)}{F(1)}-x\right)g(x)dx, \qquad I_2=\int_{1-1/M}^1\left(\frac{F(x)}{F(1)}-x\right)g(x)dx.
\end{equation*}
By continuity, there is a strictly positive constant $R$ such that 
\[ \frac{F(x)}{F(1)} -x>R\quad \forall x\in[a,b]. \] 
Therefore, we have the estimate
\[
I_1\geq R N\frac{b-a}{2}.
\] 
Next, we wish to estimate $|I_2|$ from above.  For this we note that 
\[ \left| \frac{F(x)}{F(1)} - x \right| = \left| \frac{F(x) - F(1) + F(1) (1-x)}{F(1)} \right| \] 
\[ \leq \frac{|F(x) - F(1)|}{F(1)} + |1-x| \leq \frac{||f||_\infty}{M F(1)} + \frac 1 M, \quad \forall x \in [1-1/M, 1],\] 
since 
\[ |F(x) - F(1)| = \left| \int_x ^1 f(t) dt \right| \leq (1-x) ||f||_\infty.\]
Above, $||f||_\infty$ is the supremum norm of $f$, that is finite because $f$ is continuous on the compact set $[0, 1]$.  Consequently, since $0 \leq g(x) \leq 2M$ for all $x \in [1-1/M, 1]$, we estimate 
\[ |I_2| \leq \int_{1-1/M} ^1 (2M) \frac 1 M \left( \frac{||f||_\infty}{F(1)} + 1 \right) dx \leq \frac 2 M  \left( \frac{||f||_\infty}{F(1)} + 1 \right).\] 
We thus obtain the estimate 
\[ 
I_1 + I_2 \geq I_1 - |I_2| \geq R N\frac{b-a}{2} -  \frac 2 M  \left( \frac{||f||_\infty}{F(1)} + 1 \right).
\]
Eliminating $N$ using~(\ref{eq:nmsymmetry}), we obtain an estimate that can be controlled by the value of $M$, namely
\[ 
I_1 + I_2 \geq R \frac{1-1/M}{1-a-b} -  \frac 2 M  \left( \frac{||f||_\infty}{F(1)} + 1 \right).
\]
We therefore have 
\[ 
R \frac{1-1/M}{1-a-b} -  \frac 2 M  \left( \frac{||f||_\infty}{F(1)} + 1 \right)>0
\] 
\[\iff\]
\[
M>1+2\frac{1-a-b}{R}\left(\frac{\|f\|_\infty}{F(1)}+1\right).
\]
Notice that this condition on $M$ guarantees that~(\ref{eq:mcondition1app}) is satisfied.
Consequently, 
\[ I_1 + I_2 > 0 \implies \wp(g; f) > 0. \] 
\end{proof}

We can now complete the proof of Theorem 1 for the continuous model.
\subsubsection*{Proof of Theorem 1 for the continuous model} \label{si:continuous_t1} 
\begin{proof}
Assume that $\{ f_k \}_{k=1} ^n$ is a set of strategies that satisfies the hypotheses of the theorem.  We would like to prove that these comprise an equilibrium strategy.  Consider any strategy $g$ as in Definition 2.  Then if we change strategy $f_1$ to $g$, the resulting payoff $ \wp (g; f_2, \ldots, f_n) =$

\[ \frac{1}{G(1) + \sum_{k \geq 2} F_k (1)} \int_0 ^1 g(x) \left[ \int_0 ^x \left(g(t) + \sum_{k \geq 2} f_k (t) \right)dt - \int_x ^1 \left(g(t)+ \sum_{k \geq 2} f_k (t) \right)dt \right]dx.\] 
However, by the assumption on $f_k$ in the theorem and the zero-sum dynamic, we have 

\[\wp(f_k; g) \geq 0 \implies \wp(g; f_k) \leq 0 \implies \int_0 ^1 g(x) \left[ \int_0 ^x f_k (t)dt - \int_x ^1 f_k (t)dt \right]dx \leq 0, \] 

$\forall k=2, \ldots, n$. 
Therefore, since the internal competition within $g$ does not affect the payoff, we have 

\begin{equation} \label{eq:proof2} \int_0 ^1 g(x) \left[ \int_0 ^x \sum_{k=2} ^n f_k (t)dt - \int_x ^1 \sum_{k=2} ^n f_k (t)dt \right]dx \leq 0 \implies \wp (g; f_2, \ldots, f_n) \leq 0. \end{equation} 

By the assumptions of the theorem we have for all pairs $f_j$ and $f_k$ both 

\[ \wp(f_k; f_j) \geq 0, \quad \wp(f_j; f_k) \geq 0. \] 
Due to the zero sum dynamic

\[ \wp(f_k; f_j) = - \wp(f_j; f_k) \implies \wp(f_k; f_j) = 0 \quad \forall j, k \in \{1, \ldots, n\}.\] 
Consequently, we have

\[ \int_0 ^1 f_k (x) \left[ \int_0 ^x f_\ell(t)dt - \int_x ^1 f_\ell(t)dt \right]dx = 0, \quad \forall k, \ell \in \{1, \ldots, n\},\] 
and therefore 

\[ \wp(f_k; f_1, \ldots, f_{k-1}, f_{k+1}, f_n) = 0 \quad \forall k=1, \ldots, n.\] 
By \eqref{eq:proof2}, we therefore have $\wp(g; f_2, \ldots, f_n) \leq \wp(f_1; f_2, \ldots, f_n) = 0$.   The same argument applies to $f_k$ for all $k$, namely, for any $g$ as in, Definition 2, 

\[ \wp(g; f_1, \ldots, f_{k-1}, f_{k+1}, \ldots f_n) \leq 0 = \wp(f_k; f_1, \ldots, f_{k-1}, f_{k+1}, \ldots, f_n).\] 
It follows that the set of strategies $\{f_k\}_{k=1} ^n$ is an equilibrium strategy.  

To prove the converse, we begin by assuming that $\{f_k\}_{k=1} ^n$ is an equilibrium strategy.  For the sake of contradiction, let us see what would happen if 

\[ \wp(f_1; f_2 \ldots, f_n) < 0 \implies \int_0 ^1  f_1 (x) \left[ \int_0 ^x \sum_{k=2} ^n f_k (t)dt - \int_x ^1\sum_{k=2} ^n f_k (t)dt \right]dx < 0.\] 
Defining a new strategy $g$ so that 

\[ g(x) := \sum_{k=2} ^n f_k (x) \implies \mca(g) \leq \frac 1 2,\]
and 

\[ \int_0 ^1 g(x) \left[ \int_0 ^x \sum_{k=2} ^n f_k (t)dt - \int_x ^1 \sum_{k=2} ^n f_k (t)dt \right]dx = 0 \] 

\[ \implies \wp(g; f_2, \ldots, f_n) = 0 > \wp (f_1; \ldots, f_n),\] 
contradicting the definition of equilibrium strategy.  We therefore have that $\wp(f_1; f_2, \ldots, f_n) = 0$, and the same argument shows that $\wp (f_k; f_1, \ldots, f_{k-1}, f_{k+1}, \ldots f_n ) = 0$ for all $k=1, \ldots, n$.  By the definition of equilibrium strategy, we must have that for any strategy $g$, 

\[ \wp (g; f_2, \ldots, f_n) \leq \wp (f_1; f_2 \ldots, f_n) = 0 \implies \int_0 ^1 g(x) \left[ \int_0 ^x \sum_{k=2} ^n f_k (t)dt - \int_x ^1 \sum_{k=2} ^n f_k (t)dt \right]dx \leq 0.\] 
We define a new strategy 

 \[ g_2 (x) := \sum_{k=2} ^n f_k (x) \implies \mca(g_2) \leq \frac 1 2.\] 
Then the payoff to any strategy $g$ in competition with $g_2$ satisfies 

\[ \wp (g; g_2) \leq 0.\] 
By the zero sum dynamic, we therefore have 

\[ \wp (g_2; g) \geq 0.\] 
By the definition of equilibrium strategy, since $\wp (f_k; f_1, \ldots , f_{k-1}, f_{k+1}, \ldots, f_n) = 0$ for all $k$, we define strategies 

\[ g_k (x) := \sum_{\ell \neq k} f_\ell (x),\] 
that all satisfy 

\[ \wp (g_k; g) \geq 0\] 
for all strategies $g$ as in Definition 2.  We therefore have that 

\[ \wp(g_k; g) \geq 0, \] 
with equality if and only if $\mca(g) = \frac 1 2$.  Consequently, we obtain that for all $g$ with $\mca(g) = \frac 1 2$, 

\[ \wp(g; g_k) = 0 \implies \sum_{k=1} ^n \wp(g; g_k) = 0 \] 

\[ \implies (n-1) \int_0 ^1 g(x) \left[ \int_0 ^x \sum_{k=1} ^n f_k (t)dt - \int_x ^1 \sum_{k=1} ^n f_k (t)dt \right]dx = 0,\] 
as well as 

\[ \wp(g; g_1) = 0 \implies \int_0 ^1 g(x) \left[ \int_0 ^x \sum_{k=2} ^n f_k (t)dt - \int_x ^1 \sum_{k=2} ^n f_k (t)dt \right]dx = 0.\] 
We therefore obtain that for all $g$ with $\mca(g) = \frac 1 2$, 

\[ \int_0 ^1 g(x) \left[ \int_0 ^x f_1 (t)dt - \int_x ^1 f_1 (t)dt \right]dx = 0 \implies \wp(g; f_1) = 0.\] 
It follows that in fact for all strategies $g$ as in Definition 2 we have 

\[ \wp(f_1; g) \geq 0, \] 
with equality if and only if $\mca(g) = \frac 1 2$.  The same argument shows that each $f_k$ satisfies 

\[ \wp(f_k; g) \geq 0\] 
for all $g$ as in Definition 2.  
\end{proof} 

\subsubsection*{Reproducing competitive exclusion.}  \label{ss:extinction}
\begin{figure}[h]
\centering
\includegraphics[width=0.5\textwidth]{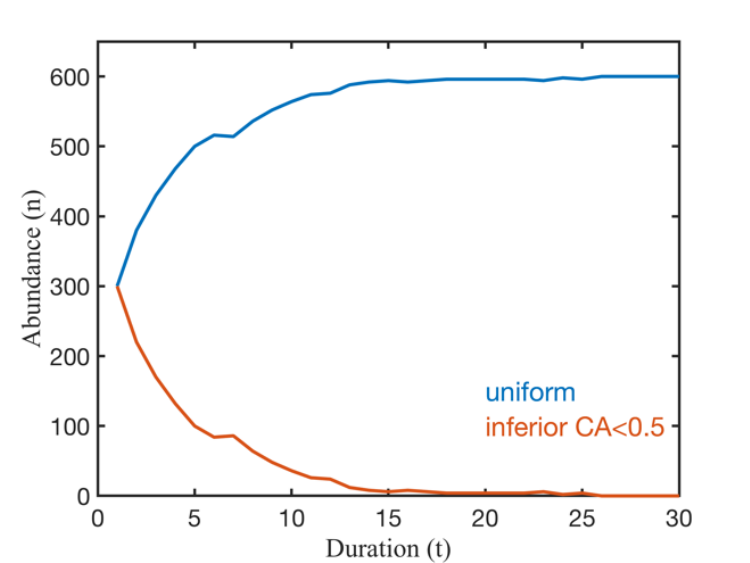}
\caption{Competition between uniform (blue) and an inferior competitor (orange) that has mean competitive ability $<0.5$. In this case, our model faithfully reproduces the ‘competitive exclusion’ principle \cite{hardin1960}, and the species with the lower mean is eliminated.} 
\label{fig:s2} 
\end{figure} 

\begin{figure}[h]
\centering
\includegraphics[width=\textwidth]{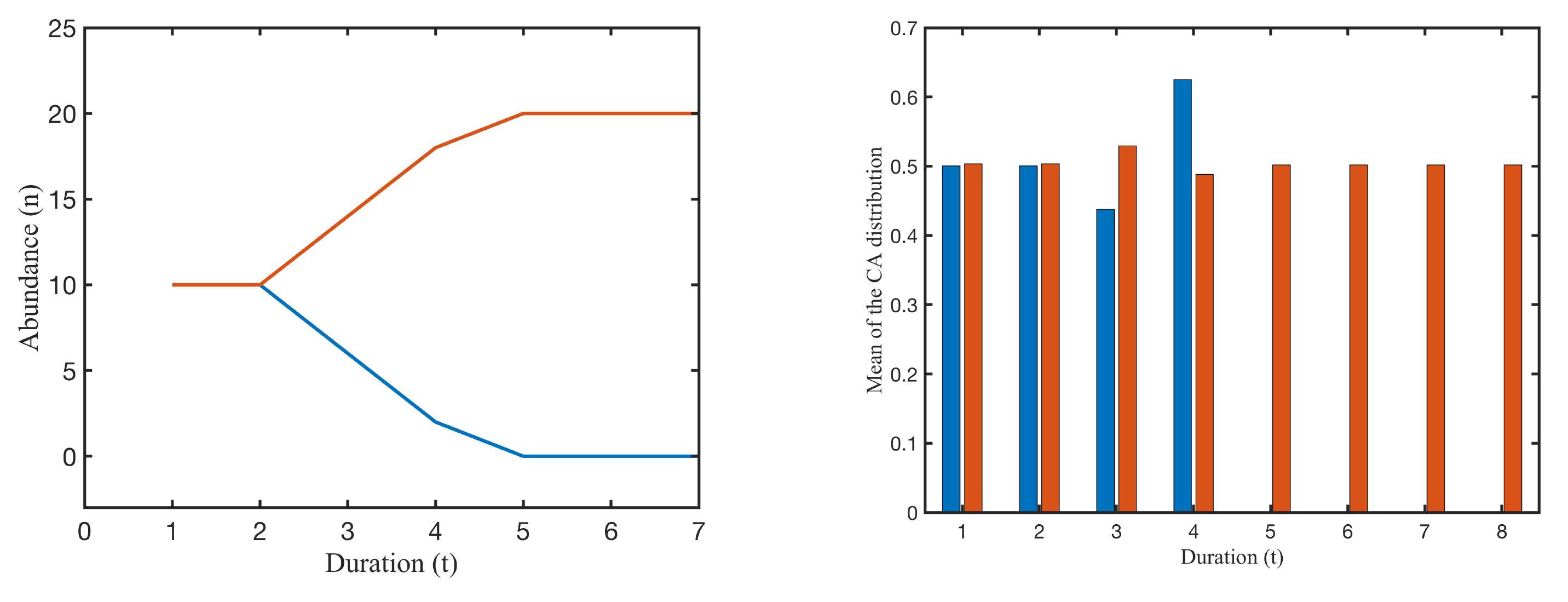}  
\caption{On the left, the nominally uniform distribution goes extinct when the uniform distribution is statistically poorly characterized because few individuals represent the distribution. The right shows the mean competitive ability of the two cohorts in this simulation. The inferior mean in time step 3 results in extinction of the blue type. Even a reversal of fortunes in the next time step, with a higher mean CA for blue than orange does not balance the competition, and the blue population is eliminated. The blue-uniform strategy could have persisted longer by chance alone, but the randomly selected orange competitors all had superior competitive ability.} 
\label{fig:s3}
\end{figure}

The trivial example of competition with an inferior competitor, characterized by a lower mean competitive ability, leads to extinction as shown in Fig \ref{fig:s2}.  At low population sizes the underlying strategy distributions can be poorly represented, especially for the maximally variable uniform distribution as shown in Fig \ref{fig:s3}. In these cases, extinctions can be observed. Such dynamics would apply when few individuals colonize a new habitat, mutations yield new functions or other events when a few individuals represent the entirety of the population.

\clearpage

%
%
%

\end{document}